\documentclass[11pt]{article}

\usepackage[letterpaper,top=2.4cm,bottom=2.4cm,left=2.6cm,right=2.6cm]{geometry}
\usepackage[T1]{fontenc}
\usepackage{amsmath,amssymb,amsthm}
\usepackage{zed-csp}          
\usepackage{tikz}
\usetikzlibrary{calc,arrows,arrows.meta,positioning}
\usepackage[most]{tcolorbox}
\usepackage{booktabs}
\usepackage{enumitem}
\usepackage{xcolor}
\usepackage[colorlinks=true,allcolors=blue]{hyperref}
\usepackage{xurl}

\theoremstyle{plain}
\newtheorem{lemma}{Lemma}
\newtheorem{theorem}{Theorem}

\newcommand{\owner}{\mathrm{owner}}
\newcommand{\bal}{\mathrm{bal}}

\definecolor{humanblue}{RGB}{0,90,160}
\definecolor{aigreen}{RGB}{0,120,0}
\newtcolorbox{humanprompt}{colback=humanblue!4,colframe=humanblue,
  title=Human Prompt,fonttitle=\bfseries\small,sharp corners,boxrule=0.8pt,
  left=5pt,right=5pt,top=3pt,bottom=3pt,breakable}
\newtcolorbox{aireply}{colback=aigreen!4,colframe=aigreen,
  title=AI Reply,fonttitle=\bfseries\small,sharp corners,boxrule=0.8pt,
  left=5pt,right=5pt,top=3pt,bottom=3pt,breakable}
\newtcolorbox{promptexcerpt}{colback=humanblue!4,colframe=humanblue,
  title=Human prompt (excerpt),fonttitle=\bfseries\small,sharp corners,boxrule=0.8pt,
  left=5pt,right=5pt,top=3pt,bottom=3pt,breakable}

\title{\textbf{Trust the Spec, Not the Code}\\[2pt]
\large A Specification-First, AI-Assisted Case Study in Online Banking}
\author{Eitan Farchi\thanks{IBM Research. Email: \texttt{farchi@il.ibm.com}.}}
\date{\today}

\begin{document}
\maketitle

\begin{abstract}
\noindent
Formal specification promises early error detection, explicit invariants, and
correctness by design, yet its notational cost has kept it out of mainstream
practice. We argue that AI removes much of that cost: natural language enriched
with lightweight mathematics, written in \LaTeX, can serve as an intermediate
specification language that is precise enough to reason over and prove, while a
large language model (LLM) reviews it for ambiguity, drafts proofs, and generates
the implementation. The specification becomes the artifact one authors, reviews,
proves, and refines; the code becomes regenerable output.
This paper is a follow-on to a prior study that established the discipline on an
organizational-knowledge-growth simulation~\cite{predecessor2026}. Here we
replicate the discipline in a different domain---an online-banking fund-transfer
service---and extend it. The two domains share one spine: a
\emph{conservation invariant} (knowledge in the prior study, money here), which
suggests the approach generalizes across domains. We contribute: (i) a
second, independent case study of the method; (ii) a stress-test of the method on a
richer problem---\emph{scheduled/recurring} transfers---whose generated code grows
substantially while the invariant and its proof do not; (iii) an \emph{AI-proposed runtime
coverage model} for invariants (``never violated $\neq$ covered''); and (iv) a Z
formalization,
including paired success/failure operation schemas and an invariant proved over
the inductive set of all reachable configurations, together with an experiment in
which the AI proposes the Z interfaces itself. We are explicit about the method's
limits: the proofs and runtime checks live at the specification level and do not
establish that the generated code refines the specification---that step is
delegated to the AI. This is a case study, not a controlled experiment.
\end{abstract}


\section{Introduction}
The benefits of formal specification have never been seriously in doubt: design
flaws surface before code exists, assumptions are written down as explicit
invariants rather than buried in implementation, and one can reason about a system
instead of debugging it into shape. What kept these benefits out of industry was
cost---heavyweight notations such as Z, VDM, and B demand mathematical maturity to
write and to maintain, the pool of engineers who can author or even read such
specifications is small, and so teams ship informal prose and defer correctness to
testing and debugging~\cite{spivey1992z,woodcock1996using,jones1990vdm,abrial1996bbook}.

Large language models change this cost structure. Modern LLMs can read
specifications written in natural language augmented with lightweight mathematics,
flag ambiguities and inconsistencies, propose invariants and proofs, and generate
running code~\cite{chen2021codex,austin2021program}. This enables a shift in what
the engineer treats as primary: the \emph{specification} becomes the artifact one
authors, reviews, proves, and refines, while the \emph{code} becomes regenerable
output that need not be read line by line. Validation shifts left---design flaws
are caught in the specification, where they are cheap, rather than in the code,
where they are not.

A prior study introduced this discipline and applied it to a simulation of
organizational knowledge growth~\cite{predecessor2026}. A natural question is
whether the discipline transfers to a different domain with genuinely different
subject matter. This paper answers that question with a second case study in
online banking. It also extends the \emph{method} in two ways---a runtime coverage
model for the central invariant and a Z formalization of the interfaces---and
stress-tests it against a richer \emph{problem}, scheduled/recurring transfers,
where the operation grows substantially while the invariant and its proof do not.
Crucially, both domains turn out to rest on the same structural device---a
conservation invariant---which is the clearest evidence we have that the approach
is not tied to one problem.

\paragraph{Research questions.}
\begin{description}[leftmargin=2.2em,itemsep=2pt,topsep=3pt]
  \item[RQ1.] Does the specification-first, AI-assisted discipline transfer to a
  new domain (online banking) with a different conserved quantity?
  \item[RQ2.] Does the discipline scale as the specification grows in operational
  complexity (from one-time to scheduled/recurring transfers)?
  \item[RQ3.] What does it take to establish, and to \emph{cover} at runtime, the
  central invariant---and where are the honest limits of the confidence obtained?
\end{description}

\paragraph{Contributions (stated as a delta over~\cite{predecessor2026}).}
The first two contributions concern \emph{scope} (a new domain and a richer
problem); the last two are \emph{methodological additions} beyond the predecessor's
specify--prove--assert loop.
\begin{enumerate}[leftmargin=1.8em,itemsep=2pt,topsep=3pt]
  \item A second, independent case study of the discipline in a new domain, giving
  first evidence of \emph{cross-domain generalization}: the conservation-invariant
  spine recurs with money in place of knowledge.
  \item A \emph{stress-test of the method} on a richer problem---%
  \emph{scheduled/recurring} transfers---in which the generated code grows markedly
  while the invariant and its proof are unchanged.
  \item \emph{Methodological addition 1:} a \emph{runtime coverage model} for
  invariants---proposed by the AI when asked what to check, then reviewed and found
  adequate---that distinguishes ``never observed to fail'' from ``all
  invariant-relevant paths exercised.''
  \item \emph{Methodological addition 2:} a \emph{Z interface layer}, with paired
  success/failure ($\Xi$) schemas and an invariant proved over the inductive set of
  all reachable configurations, plus an experiment in which the AI proposes the Z
  interfaces from the natural-language specification.
\end{enumerate}

\section{Background and Relationship to Prior Work}
\paragraph{Formal specification and Z.}
Z is a state-based specification notation built on typed set theory and
first-order logic; a \emph{schema} groups a piece of state (or a state transition)
with the predicates that constrain it, and at the code level corresponds to a typed
interface with invariants and operation contracts~\cite{spivey1992z,woodcock1996using}.
We use Z for the interface layer because pre/postconditions and state invariants
become explicit and reviewable, while the AI is free to implement the operation in
any way it chooses so long as the contract is met.

\paragraph{Conservation invariants.}
A conservation invariant asserts that some global quantity is preserved by every
operation. Such invariants are attractive targets for this discipline: they are
easy to state, provable by induction over operations, and cheap to assert at
runtime. The prior study used conservation of knowledge items; here the conserved
quantity is money.

\paragraph{The authoring flow.}
The discipline proceeds along a pipeline: \emph{natural} language $\rightarrow$
\emph{natural\_math} (natural language augmented with lightweight mathematics)
$\rightarrow$ \emph{proof} of the key invariant $\rightarrow$ \emph{Z} interfaces
$\rightarrow$ generated \emph{code}. Reviewing the proof appears to focus the
ambiguity feedback, tightening the specification before any code exists. (In this
study the Z interfaces were in fact formalized after the code; see
Sections~\ref{sec:method} and~\ref{sec:z}.)

\paragraph{Relationship to the predecessor.}
This paper is a companion to~\cite{predecessor2026}, which established the method
and its motivation on an organizational-knowledge-growth simulation. We do not
re-argue that motivation or re-describe that example; we cite it and focus on what
is new. Table~\ref{tab:diff} summarizes the differences.

\begin{table}[t]
\centering
\renewcommand{\arraystretch}{1.25}
\begin{tabular}{@{}p{0.24\textwidth} p{0.32\textwidth} p{0.34\textwidth}@{}}
\toprule
\textbf{Dimension} & \textbf{Predecessor~\cite{predecessor2026}} & \textbf{This paper} \\
\midrule
Domain & Organizational knowledge growth & Online banking (fund transfer) \\
Conserved quantity & Knowledge items & Money (sum of account balances) \\
Operations & Single simulation & One-time + scheduled/recurring transfers \\
Proof scope & Per-stage invariant & Inductive set of all reachable configurations \\
Coverage model & --- & Runtime coverage (never-violated $\neq$ covered) \\
Z formalization & Limited & Paired success/fail ($\Xi$) schemas; AI-suggested interfaces \\
\bottomrule
\end{tabular}
\caption{Differentiation from the predecessor study.}
\label{tab:diff}
\end{table}

Beyond the immediate predecessor, the work sits alongside classical formal-methods
practice~\cite{spivey1992z,woodcock1996using,jones1990vdm,abrial1996bbook} and
recent work on LLM-based code generation~\cite{chen2021codex,austin2021program};
it differs from automated program verifiers such as Dafny~\cite{leino2010dafny} in
that we do not mechanically verify the generated code---we verify the
\emph{specification} and treat the code as regenerable output (see
Section~\ref{sec:threats}).

\section{Method: One Loop, Applied to Every Problem}
\label{sec:method}
The discipline is a single loop applied to each problem:
\begin{enumerate}[leftmargin=1.8em,itemsep=2pt,topsep=3pt]
  \item \textbf{AI reviews the specification}, flagging ambiguities and
  inconsistencies before any code exists.
  \item \textbf{The human draws the line}, resolving each ambiguity and deciding
  what is controlled versus left to the AI.
  \item \textbf{Prove the invariant} (here, conservation), by induction over
  operations.
  \item \textbf{Turn the invariant into a runtime check} that is asserted as the
  simulation runs.
  \item \textbf{Fix the interface in Z}, pinning state and operations down as
  precise contracts.
  \item \textbf{Generate the code} and leave it uninspected; if the specification
  is right, the code is regenerable.
\end{enumerate}
The list above is the \emph{prescribed} order, in which the Z interfaces are fixed
before code is generated. This case study was in fact conducted \emph{code-first}:
the natural\_math specification, proof, and runtime assertion came first
(Sections~\ref{sec:onetime}--\ref{sec:coverage}), and the Z formalization was added
afterward as a separate pass (Section~\ref{sec:z}) that tightened and re-verified
the contracts. We report the study in that historical order for faithfulness, but
recommend the prescribed Z-before-code order going forward (see
Section~\ref{sec:z} and the conclusion); the mismatch is itself an observation
about how the discipline tends to be practised before it is internalized.
The remainder of the paper applies this loop twice (one-time and scheduled
transfers), adds the coverage model, and develops the Z layer.

\section{Case Study A: One-Time Transfer}
\label{sec:onetime}
\paragraph{Specification (natural\_math).}
This case study elaborates an existing online-banking specification---the
\emph{Instantpay} mobile-banking requirements document~\cite{instantpay-srs}---and
in particular its fund-transfer requirement 3.2.1.7 (functional requirement~1.7).
We begin with the simplest case, a one-time transfer; Case Study~B
(Section~\ref{sec:scheduled}) then develops the further transfer types the same
requirement calls for. We make the requirement precise as a \emph{natural\_math}
specification, as follows.
We are given a set of users $U=\{u_1,\dots,u_n\}$ and a set of accounts
$AC=\{ac_1,\dots,ac_k\}$, together with an ownership function
$AC(\cdot):U\rightarrow \mathcal{P}(AC)$ assigning to each user the accounts it
owns. Ownership partitions the accounts: for $u_i\neq u_j$,
$AC(u_i)\cap AC(u_j)=\emptyset$ and $\bigcup_{u\in U}AC(u)=AC$, so every account
belongs to exactly one user. A balance function $B:AC\rightarrow\mathbb{N}$ gives
each account a non-negative balance.

The one-time transfer accepts $(u_i,ac_i,u_j,ac_j,f)$ with $ac_i\in AC(u_i)$,
$ac_j\in AC(u_j)$, and amount $f$. If $B(ac_i)-f\ge 0$ it updates
$B(ac_i)\mapsto B(ac_i)-f$ and $B(ac_j)\mapsto B(ac_j)+f$; otherwise it does
nothing.

\begin{promptexcerpt}\small
Generate a Python program that implements the above. Simulation: three users, six
accounts (user 1 owns accounts 1--2, user 2 owns 3--4, user 3 owns 5--6), initial
balances chosen randomly; at each step choose a random transfer and report inputs
and result. (Excerpt; the full prompt and the AI's reply appear in
Appendix~\ref{app:original}.)
\end{promptexcerpt}

\noindent The generated program is 106 lines of Python
(\href{https://colab.research.google.com/drive/1HOt6whJp9vp1MiJTyrfYD0FyLORHBRNm?usp=sharing}{generated code}); inspection of the
execution log suggested it was correct, but we increase confidence by proving and
then asserting an invariant rather than by reading the code.
The simulation proceeds in discrete \emph{stages}; each stage selects one transfer
and either executes it (if funds suffice) or leaves the state unchanged. In
Case~Study~B (Section~\ref{sec:scheduled}) a clock is added, advancing by one time
unit per stage.

\paragraph{Ambiguities surfaced by review.}
The target property is a \emph{conservation invariant}: the total amount of money
held across all accounts, $\sum_{ac\in AC} B(ac)$, is unchanged by any transfer.
Asked to prove this invariant, the AI first surfaced issues that had to be resolved
for even the \emph{statement} to be well-defined:
\begin{enumerate}[leftmargin=1.8em,itemsep=1pt,topsep=3pt]
  \item \textbf{Summation notation.} $\sum_{ac_i\in AC}ac_i$ is ill-formed; the
  summand must be the \emph{balance} $B(ac_i)$, not the account identifier.
  \item \textbf{Negative balances.} The specification does not say whether a
  transfer may drive a balance negative; the proof holds either way, as it uses
  only the cancellation $-f+f=0$.
  \item \textbf{General vs.\ simulation state.} The lemma must hold for arbitrary
  $AC$ and $B$, not merely the fixed six-account configuration.
  \item \textbf{Meaning of ``amount of money.''} The informal phrase in the source
  specification was pinned down, at the AI's prompting, to the sum of the values of
  the balance function $B$---the definition adopted in the invariant stated above.
\end{enumerate}

\paragraph{The invariant, proved.}
\begin{lemma}[Conservation of total money]
\label{lem:money}
Let $B:AC\rightarrow\mathbb{N}$ be the balance function before a one-time transfer
of amount $amt$ from account $srcA$ to account $tgtA$, and let $B'$ be the updated
balance function with
$B'(a)=B(a)-amt$ if $a=srcA$, $B'(a)=B(a)+amt$ if $a=tgtA$, and $B'(a)=B(a)$
otherwise. Then $\sum_{a\in AC}B'(a)=\sum_{a\in AC}B(a)$.
\end{lemma}
\begin{proof}
Summing $B'$ over $AC$ and isolating the two affected accounts,
\begin{align*}
\sum_{a\in AC}B'(a)
&=\sum_{\substack{a\in AC\\ a\neq srcA,tgtA}} B(a)
 +\bigl(B(srcA)-amt\bigr)+\bigl(B(tgtA)+amt\bigr)\\
&=\sum_{a\in AC}B(a)-amt+amt
 =\sum_{a\in AC}B(a).
\end{align*}
The proof does not rely on $B(srcA)\ge amt$: a rejected transfer leaves $B'=B$, so
the lemma holds for aborted transfers as well. Fees or interest would require a
modified statement.
\end{proof}

\paragraph{Runtime assertion.}
The proved equality is turned into a live oracle: the simulation computes
$S_{\text{before}}=\sum_{ac}B(ac)$ and $S_{\text{after}}$ around each transfer and
asserts $S_{\text{before}}=S_{\text{after}}$. Figure~\ref{fig:ledger} shows the
conserved total for a representative transfer.

\begin{figure}[t]
\centering
\begin{tikzpicture}[font=\small]
\node[anchor=west,font=\bfseries\footnotesize,gray] at (0,3.05) {BEFORE};
\node at (-0.15,2.55) {A}; \fill[teal!70!black] (0.3,2.4) rectangle (3.3,2.72);
\node[anchor=west] at (3.4,2.56) {\$100};
\node at (-0.15,2.0) {B}; \fill[teal!55!black] (0.3,1.85) rectangle (0.9,2.17);
\node[anchor=west] at (1.0,2.01) {\$20};
\node[anchor=west,font=\bfseries] at (0,1.5) {$\Sigma=\$120$};
\node[anchor=west,font=\bfseries\footnotesize,gray] at (6,3.05) {AFTER (transfer \$40, A$\to$B)};
\node at (5.85,2.55) {A}; \fill[teal!70!black] (6.3,2.4) rectangle (8.1,2.72);
\node[anchor=west] at (8.2,2.56) {\$60};
\node at (5.85,2.0) {B}; \fill[teal!55!black] (6.3,1.85) rectangle (8.1,2.17);
\node[anchor=west] at (8.2,2.01) {\$60};
\node[anchor=west,font=\bfseries] at (6,1.5) {$\Sigma=\$120$};
\draw[->,thick,gray] (4.4,2.2)--(5.7,2.2);
\end{tikzpicture}
\caption{The conserved total is unchanged by a transfer (illustrative two-account
slice). The invariant $\sum_{ac}B(ac)$ is asserted before and after every transfer.}
\label{fig:ledger}
\end{figure}
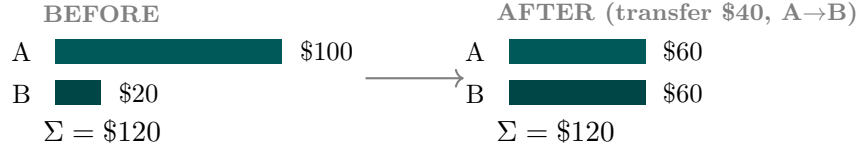

\section{Case Study B: Scheduled Transfers}
\label{sec:scheduled}
We extend the specification with time and recurrence, developing the additional
transfer types that requirement 3.2.1.7~\cite{instantpay-srs} calls for beyond the
one-time case. We keep the stage-based simulation loop of Case~Study~A and attach a
clock to it: a current time $t\in\mathbb{N}$, starting at $0$ and advancing by one
unit at the start of each stage, so that stage number and clock value coincide. In
addition to the one-time transfer, a \emph{planned transfer}
$p=(u_i,ac_i,u_j,ac_j,f,t')$ with $t'>t$ is placed in a planned-transfer set $P$;
when the clock reaches $t'$ (that is, at the stage with time $t'$), the
corresponding one-time transfer is executed and
$p$ is removed from $P$. The same stage also admits \emph{recurring} transfers---a
transfer repeated every $k$ time units---which require no new machinery: a periodic
transfer is simply a sequence of ordinary one-time transfers, so the conservation
lemma applies to each occurrence unchanged. Figure~\ref{fig:timeline} shows the
lifecycle.

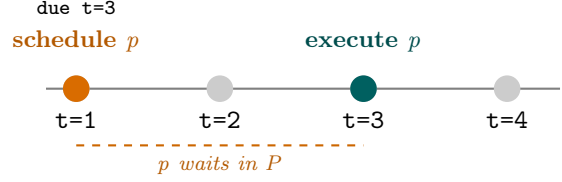
\begin{figure}[t]
\centering
\begin{tikzpicture}[font=\small]
\draw[gray,thick] (0.6,0)--(7.4,0);
\foreach \x/\n in {1/1,2.9/2,4.8/3,6.7/4}{
  \fill[gray!40] (\x,0) circle (5pt);
  \node[below=5pt] at (\x,0) {\ttfamily t=\n};}
\fill[orange!85!black] (1,0) circle (5pt);
\fill[teal!75!black] (4.8,0) circle (5pt);
\node[above=6pt,orange!70!black,font=\bfseries\footnotesize] at (1,0.15) {schedule $p$};
\node[above=20pt,font=\scriptsize] at (1,0.15) {\ttfamily due t=3};
\node[above=6pt,teal!60!black,font=\bfseries\footnotesize] at (4.8,0.15) {execute $p$};
\draw[orange!80!black,dashed,thick] (1,-0.75)--(4.8,-0.75);
\node[below,font=\scriptsize\itshape,orange!70!black] at (2.9,-0.75) {$p$ waits in $P$};
\end{tikzpicture}
\caption{Scheduled-transfer lifecycle: $p$ is scheduled at $t{=}1$ for $t{=}3$,
persists in $P$, and executes when due.}
\label{fig:timeline}
\end{figure}

\paragraph{The invariant does not change.}
The generated program grows to 265 lines
(\href{https://colab.research.google.com/drive/1HOt6whJp9vp1MiJTyrfYD0FyLORHBRNm#scrollTo=U-7j-fy4iwil&line=265&uniqifier=1}{generated code}, about $2.5\times$ the one-time
version), yet Lemma~\ref{lem:money} still applies: a planned transfer, when
executed, is an ordinary one-time transfer, and scheduling merely modifies $P$
without touching balances. The conservation invariant and its runtime assertion
carry over unchanged---the operation grew, the guarantee did not. Section~\ref{sec:z}
strengthens this into a proof over all reachable configurations.

\section{A Runtime Coverage Model for the Invariant}
\label{sec:coverage}
Asserting the invariant is necessary but not sufficient: an invariant that is never
\emph{observed} to fail may simply never have been exercised on the paths that
could break it. We therefore distinguish \emph{never violated} from
\emph{covered}: coverage requires that all semantically relevant execution paths
affecting the invariant are exercised and checked. We did not design the coverage
criterion ourselves. Prompted only with the question of what should be checked to
know the invariant is covered at runtime, the AI proposed the model below; a
subsequent AI-generated implementation then instrumented these obligations and
reported the coverage attained. Both prompts---the one that elicited the model and
the one that generated the checking code---appear in full in
Appendix~\ref{app:original}; the first was, in essence:
\begin{promptexcerpt}\small
Given the invariant [$\sum_{ac\in AC}B(ac)$ is constant], what should we check to
determine that the invariant was covered at runtime?
\end{promptexcerpt}
\noindent The obligations the AI returned were:
\begin{enumerate}[leftmargin=1.8em,itemsep=2pt,topsep=3pt]
  \item \textbf{Pre/post checking.} For every executed transfer (one-time or
  planned), compute the total $S_{\text{before}}=\sum_{ac\in AC}B(ac)$ immediately
  before the transfer and $S_{\text{after}}=\sum_{ac\in AC}B(ac)$ immediately after,
  and assert $S_{\text{before}}=S_{\text{after}}$.
  \item \textbf{Both semantic cases.} Exercise the success case
  ($ac_s-f\ge 0$) and the failure case ($ac_s-f<0$) at least once for each
  transfer type.
  \item \textbf{User and account diversity.} Observe transfers across different
  users and distinct accounts, exercising ownership disjointness.
  \item \textbf{Planned-transfer lifecycle.} Observe insertion into $P$ (no balance
  change), persistence across time steps, and execution when due, asserting
  $S_{\text{before}}=S_{\text{after}}$ immediately before and after that execution.
  \item \textbf{Multi-transfer stages.} Include at least one stage with several
  transfers, checking the invariant after each.
\end{enumerate}
The invariant is \emph{covered} when all five obligations are met during execution.

\paragraph{Reviewing the coverage model.}
We do not take the AI's proposal on trust. As with the specification and the
AI-proposed invariants, \emph{reviewing} the proposed coverage model is itself part
of the discipline. On review we judged the model adequate: the five obligations
together exercise every path that could break conservation---both transfer outcomes
(success and failure), ownership diversity across users and accounts, the full
planned-transfer lifecycle, and stages containing several transfers---so no
invariant-relevant behaviour is left unchecked. The AI-generated implementation then
instrumented these obligations and reported $100\%$ coverage.

This coverage model---AI-proposed, human-reviewed, then checked by AI-generated
code---is one of two methodological additions this paper makes beyond the
predecessor's runtime-assertion approach; the other is the Z interface layer of
Section~\ref{sec:z}, which the predecessor only touched on. The generated
implementation that performs the invariant check and attains full coverage is
available online
(\href{https://colab.research.google.com/drive/1eB4d7n5-At23hP1QXLmpwKHiSY-PwWTL#scrollTo=3hMqSpJeS4eO&line=208&uniqifier=1}{generated code}).

\section{Z Formalization of the Interfaces}
\label{sec:z}
We fix the interface layer in Z. The schemas in this section were \emph{generated
by the AI}: prompted to rewrite the natural\_math specification so that the
fund-transfer and scheduled-transfer interfaces are made explicit in Z schema
notation (leaving the rest of the specification untouched), the AI produced the
state and operation schemas below; we review and lightly edit them, exactly as we
review the specification, the proposed invariants, and the coverage model. In this
study they were generated \emph{after} the code---a formalization pass that made the
contracts precise and let us restate the invariant as a property of the
specification itself (Theorem~\ref{thm:inductive}). With hindsight we would generate
them \emph{before} the code, as the method prescribes: fixing the contract first
removes ambiguity earlier and gives the generator a precise target.

\paragraph{Z as an internal representation, not a user-facing notation.}
We do not advocate Z as the notation the user reads or writes. Its role here is
\emph{internal}: the schema mechanism decomposes the specification into named
interfaces with explicit inputs, invariants, and pre/postconditions, which both
sharpens ambiguity detection and---by splitting the problem into independent
pieces---mitigates the attention limits that degrade generation at scale. What the
user sees can remain the lightweight ``natural language $+$ mathematics'' form:
whenever an ambiguity surfaces, or an interface itself needs to be presented, it can
be rendered back into that readable form. Z works behind the scenes; \emph{natural
math} stays the human-facing surface.

Account balances are held in a state schema; the
planned-transfer store is a second piece of state; each operation is a schema with
explicit pre/postconditions. We pair a success schema with a $\Xi$
(no-change) failure schema so that the insufficient-funds case is an explicit,
reviewable contract rather than an afterthought.

\begin{schema}{AccountState}
  balances: AC \fun \nat
\end{schema}

\noindent A fixed ownership function assigns each account to its owning user:
\begin{axdef}
  owner: AC \fun U
\end{axdef}

\begin{schema}{PlannedTransferStore}
  P: \power (U \cross AC \cross U \cross AC \cross \nat_1 \cross \nat)
\end{schema}

\begin{schema}{OneTimeTransfer}
  \Delta AccountState \\
  srcU, dstU: U; \; srcAC, dstAC: AC; \; f: \nat_1
\where
  owner~srcAC = srcU \\
  owner~dstAC = dstU \\
  balances~srcAC \geq f \\
  balances' = balances \oplus \{ srcAC \mapsto balances~srcAC - f,\; dstAC \mapsto balances~dstAC + f \}
\end{schema}

\begin{schema}{OneTimeTransferFail}
  \Xi AccountState \\
  srcU, dstU: U; \; srcAC, dstAC: AC; \; f: \nat_1
\where
  owner~srcAC = srcU \\
  owner~dstAC = dstU \\
  balances~srcAC < f
\end{schema}

\begin{schema}{SchedulePlannedTransfer}
  \Delta PlannedTransferStore \\
  srcU, dstU: U; \; srcAC, dstAC: AC; \; f: \nat_1; \; t?: \nat
\where
  P' = P \cup \{ (srcU, srcAC, dstU, dstAC, f, t?) \}
\end{schema}

\noindent A further schema \textit{ExecutePlannedTransfers}$(t)$ applies every
planned transfer in $P$ due at time $t$ (each an ordinary balance update) and
removes it from $P$; we omit its display for brevity.

\paragraph{An omission caught on review.}
The AI-generated schemas initially declared the users $srcU,dstU$ but never
constrained them: nothing tied $srcAC$ to $srcU$, so the contract permitted
debiting an account that its stated user does not own. The natural\_math
specification requires the source account to be owned by the source user (and
likewise for the destination); we restored this as the preconditions
$owner~srcAC = srcU$ and $owner~dstAC = dstU$ above. The omission does not affect
conservation---money is conserved regardless of ownership, which is why the proof
never exposed it---but it does matter for a faithful interface contract. It is
exactly the kind of gap that reviewing the AI's output, rather than trusting it, is
meant to catch.

\paragraph{The invariant, over all reachable configurations.}
Let $I$ be the inductive set of configurations obtained from any legal initial
configuration by finitely many applications of the four operation schemas. The
following claim is strictly stronger than Lemma~\ref{lem:money}: it is a property
of the specification itself, independent of any particular simulation or execution
order.

\begin{theorem}[Conservation over $I$]
\label{thm:inductive}
For every configuration in $I$, $\sum_{ac\in AC}balances(ac)$ is constant.
\end{theorem}
\begin{proof}[Proof sketch]
By induction on the construction of $I$. \emph{Base:} the initial total is fixed.
\emph{Step:} \textit{OneTimeTransfer} changes two balances by $-f$ and $+f$,
preserving the sum; \textit{OneTimeTransferFail} and \textit{SchedulePlannedTransfer}
include $\Xi AccountState$ (or leave balances untouched), preserving the sum; and
\textit{ExecutePlannedTransfers} is a finite composition of balance-preserving
updates. Hence the sum is invariant across $I$.
\end{proof}

\paragraph{Letting the AI choose the interfaces.}
The schemas above were obtained by asking the AI to \emph{translate} the
natural\_math specification into Z. As a stronger probe, we instead gave it only the
natural-language specification and asked it to \emph{decide the interface
decomposition itself}. It performed well on substance---it identified the major
interfaces unprompted (one-time transfer, schedule, execute, and a time-advance
operation), captured the ownership invariant
($owner(ac_1)=owner(ac_2)\Rightarrow ac_1=ac_2$), and expressed pre/postconditions
as explicit contracts---but slipped in ways a human must catch: it interleaved
Z schemas with natural-language prose, leaked a simulation-specific indexing detail
into an interface, and quietly omitted parts of the specification such as the
initialization. The full AI-proposed interfaces, including these omissions, appear
in Appendix~\ref{app:original}. The AI can draft the interface surface; a human
still curates what belongs in the contract versus the simulation.

\section{Results and Observations}
\label{sec:results}
Across both case studies, the code was never inspected line by line as it evolved;
confidence came from specification-level review, the conservation proof, the
runtime coverage model, and the Z contracts. We report the following observations,
which the reader should treat as case-study evidence rather than controlled
measurements:
\begin{itemize}[leftmargin=1.8em,itemsep=2pt,topsep=3pt]
  \item The same conservation invariant---and its proof---carried over unchanged
  from the one-time to the scheduled variant, even as the generated code grew
  substantially. The invariant is universal to the online-banking specification,
  not tied to any particular operation.
  \item The conservation assertion held on every checked stage of every run we
  performed, as confirmed by inspecting the simulation output (viewable by
  following the links to the generated code).
  \item The implementation appeared correct on first generation once the
  specification had been reviewed---``correct'' here meaning that the logged
  simulation output, including the per-stage invariant assertion, showed the
  expected behaviour, again by inspection of the generated-code output rather than
  by reading the code.
  \item On the positive side, the AI contributed correct lower-level invariants the
  author had not stated, and proposed usable Z interfaces from prose.
  \item On the negative side, as specifications grew past a few pages the AI
  sometimes silently dropped tasks or sections---so abstraction, modularity, and
  decomposition still matter.
\end{itemize}

\section{Discussion, Limitations, and Threats to Validity}
\label{sec:threats}
\paragraph{The central gap: specification-level, not code-level, guarantees.}
The proofs (Lemma~\ref{lem:money}, Theorem~\ref{thm:inductive}) and the runtime
assertions establish properties of the \emph{specification} and of \emph{observed
executions}. They do \emph{not} establish that the generated \emph{code} refines
the specification. That refinement step is delegated to the AI and is
\emph{not} formally discharged here---unlike, e.g., mechanically verified
development~\cite{leino2010dafny}. Runtime assertions raise confidence but detect
violations only on exercised paths (hence the coverage model of
Section~\ref{sec:coverage}), and coverage itself is argued informally rather than
measured by an independent tool.

\paragraph{Case study, not controlled experiment.}
This is a single-author case study in one domain. The empirical observations in
Section~\ref{sec:results} are anecdotal and are flagged as such; there is no
control condition, no independent replication, and the same person authored the
specification and judged correctness. In particular, ``correct on first
generation'' rests on inspection of the simulation output by that same author, not
on an independent oracle, so its construct validity is weak until a measurement
protocol is specified.

\paragraph{Tooling and reproducibility.}
Results depend on a specific commercial LLM assistant. The work was carried out
with ChatGPT used as an agent during the first quarter of 2026; the exact
underlying model version is not known to us, as it was not surfaced by the
interface and may have changed over the period. Re-running with a different model or
at a different date may therefore differ.

\paragraph{Generality.}
Conservation invariants are unusually well suited to this discipline. Whether the
approach extends as cleanly to properties without a conservation structure (e.g.,
liveness, security) is open.

\section{Conclusion and Future Work}
We replicated a specification-first, AI-assisted discipline in a new domain and
extended the \emph{method} with two additions---a runtime coverage model and a Z
formalization proved over all reachable configurations---while stress-testing it
against a richer \emph{problem}, scheduled/recurring transfers, which the
online-banking requirement calls for rather than the method. The recurrence of a
single conservation-invariant spine across two very different domains---knowledge
and money---is preliminary evidence that the discipline generalizes. The honest
limit
remains that confidence is established at the specification level; closing the gap
to the code (via refinement checking or verified generation) is the natural next
step, along with controlled evaluation, additional domains, and non-conservation
properties. Two domain extensions we scoped but did not complete also point the
way: a \emph{recurring-solvency} invariant---guaranteeing that an account holds
sufficient funds before a scheduled debit falls due, a liveness/scheduling property
rather than a conservation one---and support for \emph{multiple, heterogeneous
transfer types} within a single specification. Both would test whether the
discipline holds for properties and operation sets richer than the conservation
invariant studied here. Finally, although this study was carried out code-first
with the Z interfaces formalized afterward, our experience suggests authoring the
Z contracts \emph{before} generation---fixing the interface first, then treating
the code as regenerable output beneath it---and we adopt that order as the
recommended practice. We also stress that the Z layer was itself AI-generated and is
intended as an \emph{internal} representation rather than a user-facing notation: its
schema decomposition sharpens ambiguity detection and helps generation scale, while
the human-facing surface stays the lightweight ``natural language $+$ mathematics''
form, into which any surfaced ambiguity or interface can be rendered back for review.
Whether this internal-Z decomposition measurably improves ambiguity detection and
large-scale generation is a further question we leave open. More broadly, we see the
discipline itself---specification review, invariant proof, coverage, and internal-Z
decomposition---as a \emph{pattern} that could be formalized as a reusable
\emph{skill} and then applied at scale: instantiated across different AI agents and
a range of use cases, and evaluated systematically rather than through a single
case study. Establishing such a skill, and measuring how it transfers across agents
and domains, is the direction we consider most promising.

\section*{Reproducibility and Artifacts}
The paper is designed to be reproducible from the artifacts it carries.
Appendix~\ref{app:original} reproduces the original development document, including
the full \emph{Human Prompt} / \emph{AI Reply} transcripts for every step---the
specification, the ambiguity resolutions, the invariant proof, the coverage model,
and the Z interfaces. It also retains the links to the AI-generated code (the
one-time and scheduled simulations and the coverage-checking implementation),
reproduced inline in Sections~\ref{sec:onetime}, \ref{sec:scheduled},
and~\ref{sec:coverage}. Together these let a reader repeat the process: re-issue the
same prompts to reproduce the reported results, or re-issue them---adapted---against
a different domain or a different AI agent to test how the discipline transfers.

\section*{Acknowledgments}
This article was developed from the appendix material (the specifications, prompts,
and AI responses) with the assistance of Anthropic's Claude, used to structure and
draft the manuscript. This is distinct from the development work reported in the
paper, which used ChatGPT as an agent (Section~\ref{sec:threats}). Responsibility
for the correctness of the content lies entirely with the author.

\bibliographystyle{plain}
\bibliography{main}

\appendix
\section{Original Development Document (lightly edited)}
\label{app:original}
This appendix reproduces the original working document behind this paper, retained
for reference. The English has been lightly edited for readability; the technical
content, the full \emph{Human Prompt} / \emph{AI Reply} transcripts, the
alternative ``trails'' from natural\_math to Z, and all external links (including
the generated-code notebooks) are preserved as in the original.

\bigskip
\sloppy\footnotesize
\subsection{The online bank}

We will use the following online banking specification.  See \href{https://www.researchgate.net/profile/Adya-Singh-7/publication/375757633_Software_requirements_specification_Mobile_Banking_Application_Instantpay/links/658abb6f3c472d2e8e905907/Software-requirements-specification-Mobile-Banking-Application-Instantpay.pdf?origin=publication_detail&_tp=eyJjb250ZXh0Ijp7ImZpcnN0UGFnZSI6InB1YmxpY2F0aW9uIiwicGFnZSI6InB1YmxpY2F0aW9uRG93bmxvYWQiLCJwcmV2aW91c1BhZ2UiOiJwdWJsaWNhdGlvbiJ9fQ}{link}.  We consider the fund transfer requirement (3.2.1.7 Functional requirement 1.7).  We start by a simple simulation of one time transfer.  Next, we will add assertions and then further develop the bank transfer to support additional types of transfer as required in 3.2.1.7 .

\subsection{From natural\_math to proof to code}

Here, the specification is written by the human in natural language augmented by mathematical notation and is written in Latex.  In addition, by reviewing the proof the confidence that the specification is clean increases.  The proof is possibly a mechanism of focusing the ambiguity feedback but that remains to be proven. 

\subsubsection{Fund transfer - one time transfer}

\paragraph{Fund transfer - one time transfer}
\begin{humanprompt}
\textbf{Prompt: Generate a python program that implements the following}\\

We are given a set of users $U = \{ u_1, \ldots, u_n\}$.  Each user has a set of accounts that belong to that user.  In other words we have a set of accounts $AC = \{ac_1,\ldots, ac_k\}$ and a function $AC(u): U \longrightarrow P(AC)$.  Thus, for any users $u \in U$, $AC(u) \subseteq AC$ is a subset of the accounts $AC$ that belong to $u$.   

We assume that for any two users $u_i, u_j \in U$ the intersection between $AC(u_i)$ and $AC(u_j)$ is empty.  In addition, $\bigcup_{u \in U} AC(u) = AC$. Thus, an account belongs to one and only one user.

An account $ac_i \in AC$ is a natural number (including zero).  It represents the amount of money that the user that owns the account, i.e., the user $u$ such that $ac_i \in AC(u)$.  

The one time transfer is an interface that accepts a source user, $u_i \in U$, an account $ac_i \in AC(u_i)$, and a transfer amount $f$.  It also accepts a target user $u_j \in U$ and an account $ac_j \in AC(u_j)$.  

The one time transfer preforms the following operation. If $ac_i - f \ge 0$ then the interface updates $ac_i$ to $ac_i-f$ and $ac_j$ to $ac_j+f$.  Otherwise the interface does nothing.  

The simulation of one time transfer executes as follows. 

\begin{enumerate}
    \item Randomly chose inputs for the one time transfer interface.  In other words, randomly chose  a source user, $u_i \in U$, an account $ac_i \in AC(u_i)$, and a transfer amount $f$.  Also randomly chose a target user $u_j \in U$ and an account $ac_j \in AC(u_j)$.  
    \item Perform the transfer operation 
    \item Report the input that was chosen to the transfer operation and the result of the transfer operation.
\end{enumerate}

\textbf{Simulation initialization}\\
Choose the number of users to be three and the number of accounts to be 6.  Have user one own accounts 1 and 2, user two own account 3 and 4 and user 3 own account 5 and 6. Choose the amount of money in each account randomly. 

\end{humanprompt}

\paragraph{Generate code}

Generated using ChatGPT.  See generated code here \href{https://colab.research.google.com/drive/1HOt6whJp9vp1MiJTyrfYD0FyLORHBRNm?usp=sharing}{link}.  Inspecting the execution log of the simulation indicates that the simulation is correct but we would like to increase our confidence level in the implementation by introducing assertions.  The code is $106$ LOC long. 

\paragraph{Assertion - the overall money in the system is constant}

\begin{humanprompt}

\begin{lemma}
\label{assertion}
The amount of money $\sum_{ac_i \in AC} ac_i$ remains constant after each fun transfer. 
\end{lemma}

\textbf{Provide a proof that lemma \ref{assertion} is correct}\\

\end{humanprompt}

Using ChatGPT I obtained the following proof of the assertion.  Reviewing it I found the proof correct. 

\begin{aireply}

\subparagraph{Generated proof}

\paragraph{Setup and assumptions.}
Let \(AC=\{ac_1,\dots,ac_k\}\) be the multiset of account balances
(we identify an account by its current numeric balance).
Assume every account balance and every transfer amount is a real
number (or integer) \(\ge 0\). In particular we assume the transfer
amount \(f\) satisfies \(f\ge 0\). The one-time transfer interface
selects a source account \(ac_s\) and a target account \(ac_t\)
(possibly belonging to different users) and attempts to move \(f\)
units from \(ac_s\) to \(ac_t\). The interface updates balances only
when \(ac_s - f \ge 0\); otherwise it leaves all balances unchanged.

\begin{lemma}
\label{assertion}
The total amount of money
\(\displaystyle T=\sum_{ac\in AC} ac\) is invariant under the one-time
transfer operation: after any transfer (successful or aborted)
the new total \(T'\) equals \(T\).
\end{lemma}

\begin{proof}
Let the accounts before the transfer be \(AC=\{ac_1,\dots,ac_k\}\)
and denote the chosen source and target accounts by \(ac_s\) and
\(ac_t\) respectively. Write the total before the transfer as
\[
T \;=\; \sum_{ac\in AC} ac \;=\; \Big(\sum_{ac\in AC\setminus\{ac_s,ac_t\}} ac\Big)
       + ac_s + ac_t .
\]

We consider two cases.

\textbf{Case 1 — transfer aborts:} If \(ac_s - f < 0\), the interface
does nothing. Hence every account balance remains the same and the
total after the operation is \(T' = T\).

\textbf{Case 2 — transfer succeeds:} If \(ac_s - f \ge 0\), the
interface updates the two chosen accounts to
\[
ac_s' = ac_s - f,\qquad ac_t' = ac_t + f,
\]
and leaves every other account unchanged. The total after the
transfer is
\[
T' \;=\; \Big(\sum_{ac\in AC\setminus\{ac_s,ac_t\}} ac\Big)
       + ac_s' + ac_t'
    \;=\; \Big(\sum_{ac\in AC\setminus\{ac_s,ac_t\}} ac\Big)
       + (ac_s - f) + (ac_t + f).
\]
The \(-f\) and \(+f\) cancel, so
\[
T' \;=\; \Big(\sum_{ac\in AC\setminus\{ac_s,ac_t\}} ac\Big)
       + ac_s + ac_t \;=\; T.
\]

In both cases \(T' = T\). Therefore the total amount of money in the
system is preserved by every one-time transfer operation.
\end{proof}

\end{aireply}

TBC - add recurring 
transfers, and scheduled transfers

TBC - add assertions that the some of money in the system is constant 

TBC - add assertion that recurring transfer should add money to an account before an upcoming recurring reduction in the account. 

\subsubsection{Fund transfer - schedule transfer}

Next, I updated the specification to include a transfer at time $t$ and a recurring transfer every $k$ unit of times.

\begin{humanprompt}
\label{scheudleTransfer}
\textbf{Prompt: Generate a python program that implements the following}\\

We are given a set of users $U = \{ u_1, \ldots, u_n\}$.  Each user has a set of accounts that belong to that user.  In other words we have a set of accounts $AC = \{ac_1,\ldots, ac_k\}$ and a function $AC(u): U \longrightarrow P(AC)$.  Thus, for any users $u \in U$, $AC(u) \subseteq AC$ is a subset of the accounts $AC$ that belong to $u$.   

We assume that for any two users $u_i, u_j \in U$ the intersection between $AC(u_i)$ and $AC(u_j)$ is empty.  In addition, $\bigcup_{u \in U} AC(u) = AC$. Thus, an account belongs to one and only one user.

An account $ac_i \in AC$ is a natural number (including zero).  It represents the amount of money that the user that owns the account, i.e., the user $u$ such that $ac_i \in AC(u)$.  

Fund transfer are a function of the current time. The current time is a natural number initialized to $0$ at the beginning of the simulation.  At each stage of the simulation the current time is first incremented by one. 

\paragraph{Fund transfer operations - schedule transfer}

\textbf{The one time transfer} is an interface that accepts a source user, $u_i \in U$, an account $ac_i \in AC(u_i)$, and a transfer amount $f$.  It also accepts a target user $u_j \in U$ and an account $ac_j \in AC(u_j)$.  

The one time transfer preforms the following operation. If $ac_i - f \ge 0$ then the interface updates $ac_i$ to $ac_i-f$ and $ac_j$ to $ac_j+f$.  Otherwise the interface does nothing.  

\textbf{Planed transfer} is a one time transfer scheduled to occur at time $t$.  See simulation steps below for details 

\subparagraph{Simulation}

The simulation of fund transfer consists of stages. In each stage the following steps are taken. 

\begin{enumerate}
    \item Increment current time, $t$, by $1$. 
    \item \textbf{One time transfer}.  Randomly chose inputs for the one time transfer interface.  In other words, randomly chose  a source user, $u_i \in U$, an account $ac_i \in AC(u_i)$, and a transfer amount $f$.  Also randomly chose a target user $u_j \in U$ and an account $ac_j \in AC(u_j)$.  Thus, a one time transfer is defined by the vector  $(u_i, ac_i, u_j, ac_j, f)$.
    \begin{enumerate}
           \item Perform the transfer operation $(u_i, ac_i, u_j, ac_j, f)$.  
    \end{enumerate}
       
    \item \textbf{Schedule a planned transfer}. Randomly chose inputs for a one time transfer interface.  In other words, randomly chose  a source user, $u_i \in U$, an account $ac_i \in AC(u_i)$, and a transfer amount $f$.  Also randomly chose a target user $u_j \in U$ and an account $ac_j \in AC(u_j)$.  Randomly choose some time $t^{'}$ that is greater than the current time.   Add the vector $p = (u_i, ac_i, u_j, ac_j, f, t^{'})$ in the planned transfer set $P$.
    \item \textbf{Perform a planned transfer}. For any $p = (u_i, ac_i, u_j, ac_j, f, t^{''}) \in P$ such that $t^{''} = t$ the current time perform the one time transfer $(u_i, ac_i, u_j, ac_j, f)$ and remove $p = (u_i, ac_i, u_j, ac_j, f, t^{''})$ from $P$. 
    \item Report the input that was chosen for each of the transfer operations and the result of each transfer operation performed in this stage.  Identify which type of transfer was preformed (planned or one time).  Highlight the current time of each transfer operation performed in this stage, $t$.  Report the state of the planned transfer set $P$ at that stage before the planned transfer operations were performed.  
\end{enumerate}

\textbf{Simulation initialization}\\
Choose the number of users to be three and the number of accounts to be 6.  Have user one own accounts 1 and 2, user two own account 3 and 4 and user 3 own account 5 and 6. Choose the amount of money in each account randomly. Current time is initialized to $0$. Initialize the planned fund transfer set $P$ to the empty set.  An element in $P$ is of type $U \times AC \times U \times AC \times N^{+} \times N$.  $N^{+}$ stand for the natural numbers without zero and $N$ stand for the natural numbers with $0$.

\end{humanprompt}

\paragraph{Generated code}

Using ChatGPT I generated Python code see \href{https://colab.research.google.com/drive/1HOt6whJp9vp1MiJTyrfYD0FyLORHBRNm#scrollTo=U-7j-fy4iwil&line=265&uniqifier=1}{link}.  The code is $265$ LOC and seems to be executing fine based on inspection of the simulation.  Next we also generate the assertions based on the lemma above to increase our confidence in the implementation.  But before we do that let us check that the assertion still holds. 

\paragraph{Assertion - the overall money in the system is constant}

I gave the updated design definition to ChatGPT in a new session and asked for the proof of the lemma below.   The proof is essentially the same.   The only difference is the observation that periodic fund transfer translates to regular transfer - see proof below.

\begin{humanprompt}

\begin{lemma}
\label{assertion}
The amount of money $\sum_{ac_i \in AC} ac_i$ remains constant after each fun transfer. 
\end{lemma}

\textbf{Provide a proof that lemma \ref{assertion} is correct}\\

\end{humanprompt}

\subparagraph{Generated proof}

\begin{aireply}
\begin{lemma}
The amount of money $\sum_{ac_i \in AC} ac_i$ remains constant after each fund transfer.
\end{lemma}

\begin{proof}
Let $S = \sum_{ac \in AC} ac$ denote the total amount of money before a transfer.

Consider a one-time transfer with input $(u_s, ac_s, u_t, ac_t, f)$, where $ac_s$ is the source account and $ac_t$ the target account.

\textbf{Case 1:} $ac_s - f \ge 0$.  
The transfer is executed, yielding
\[
ac_s' = ac_s - f, \qquad ac_t' = ac_t + f.
\]
All other accounts remain unchanged. The new total is
\[
S' = \sum_{ac \in AC} ac' 
= \Bigl(\sum_{ac \in AC \setminus \{ac_s, ac_t\}} ac \Bigr) + (ac_s - f) + (ac_t + f).
\]
Rearranging gives
\[
S' = \sum_{ac \in AC} ac = S.
\]

\textbf{Case 2:} $ac_s - f < 0$.  
No transfer occurs, so no balances change and $S' = S$.

Thus every one-time transfer preserves the total amount of money.

A planned transfer is executed exactly as a one-time transfer at its scheduled time, so it also preserves the total. Adding a planned transfer to the set $P$ or advancing the simulation time does not modify any account values.

Since every transfer preserves the total, and a simulation stage may contain several such transfers, the total also remains unchanged after any stage.

Therefore,
\[
\sum_{ac \in AC} ac \text{ is invariant under the simulation}.
\]
\end{proof}
\end{aireply}

\subsection{Invariant coverage}

Next, in order to increase the confidence level that the implementation is correct I'll attempt to define a coverage model of the invariant \ref{assertion} and obtain runtime feedback of the level of coverage the invariant has at simulation time. 

\subsubsection{Obtain a coverage model for the invariant}
\begin{humanprompt}
Given the invariant \ref{assertion} what should we check to determine that the invariant was coveraged at runtime? 
\end{humanprompt}

Here is the reply I got from ChatGPT.  

\begin{aireply}

\paragraph*{Runtime Coverage Model for the Invariant}

\paragraph{Invariant.}
\[
\sum_{ac \in AC} ac = \text{constant}
\]

The invariant states that the total amount of money across all accounts must remain unchanged after every fund transfer operation.

\subparagraph*{Meaning of Coverage}

Invariant coverage at runtime means more than observing that the invariant is never violated. It requires that all semantically relevant execution paths that could affect the invariant are exercised and checked during execution.

\subparagraph*{Coverage Obligations}

The invariant is considered covered if all of the following runtime conditions are satisfied.

\paragraph{1. Pre- and Post-Transfer Checking}
For every executed transfer operation (one-time or planned), the system computes
\[
S_{\text{before}} = \sum_{ac \in AC} ac
\quad \text{and} \quad
S_{\text{after}} = \sum_{ac \in AC} ac
\]
and verifies that
\[
S_{\text{before}} = S_{\text{after}}.
\]

This check must be performed for both one-time transfers and planned transfers.

\paragraph{2. Semantic Case Coverage of Transfer Execution}

Each transfer operation has two semantic cases, both of which must be exercised.

\begin{itemize}
\item \textbf{Successful transfer:} $ac_s - f \ge 0$, where balances are updated.
\item \textbf{Failed transfer:} $ac_s - f < 0$, where no balances change.
\end{itemize}

Runtime coverage requires that both cases occur at least once for one-time transfers and at least once for planned transfers.

\paragraph{3. User and Account Diversity}

Transfers must be observed across different users and distinct accounts, including cases where the source user and target user are different. This ensures that the disjointness of account ownership is exercised at runtime.

\paragraph{4. Planned Transfer Lifecycle Coverage}

Planned transfers must be observed through all phases of their lifecycle:
\begin{enumerate}
\item Insertion into the planned transfer set $P$, without modifying account balances.
\item Persistence in $P$ across multiple time steps.
\item Execution at the scheduled time, with invariant checking before and after execution.
\end{enumerate}

\paragraph{5. Multi-Transfer Stage Coverage}

At least one simulation stage must include multiple transfer operations (e.g., a one-time transfer and one or more planned transfers). The invariant must be checked and preserved after each individual transfer in the stage.

\subparagraph*{Coverage Criterion}

The invariant $\sum_{ac \in AC} ac$ is considered covered at runtime if all of the above conditions are satisfied during the execution of the simulation.

\end{aireply}

\subsubsection{Code generation}

\begin{humanprompt}
Generate the python code that is specified above.  Inlcude a check of the invariance \ref{assertion} after each stage of the simulation. Check the coverage detailed above and report the increase in coverage after each stage and what new condition was covered. At the end of the simulation report the percentage of coverage that was achieved. 
\end{humanprompt}

Generated code executed nicely and achieved $100\%$ coverage.  See code here \href{https://colab.research.google.com/drive/1eB4d7n5-At23hP1QXLmpwKHiSY-PwWTL#scrollTo=3hMqSpJeS4eO&line=208&uniqifier=1}{link}.

\subsection{From natural to natural\_math to proof to code}

Can the business rule be specified in natural language and translated to the natural\_math intermediate representation and use the claim and proof flow to increase the confidence that the specification is clean?

\subsubsection{Fund transfer - one time transfer - using a natural to natural\_math to proof to code flow}

First we create a natural language description of the specification and verify that it describes the one time transfer requirement in natural language with minimal or no use of mathematical notation. 

\begin{humanprompt}
\textbf{Prompt: Generate a natural language description of the following specification keeping mathematical notation to a minium.  Generate the specification as a latex section.}\\

We are given a set of users $U = \{ u_1, \ldots, u_n\}$.  Each user has a set of accounts that belong to that user.  In other words we have a set of accounts $AC = \{ac_1,\ldots, ac_k\}$ and a function $AC(u): U \longrightarrow P(AC)$.  Thus, for any users $u \in U$, $AC(u) \subseteq AC$ is a subset of the accounts $AC$ that belong to $u$.   

We assume that for any two users $u_i, u_j \in U$ the intersection between $AC(u_i)$ and $AC(u_j)$ is empty.  In addition, $\bigcup_{u \in U} AC(u) = AC$. Thus, an account belongs to one and only one user.

An account $ac_i \in AC$ is a natural number (including zero).  It represents the amount of money that the user that owns the account, i.e., the user $u$ such that $ac_i \in AC(u)$.  

The one time transfer is an interface that accepts a source user, $u_i \in U$, an account $ac_i \in AC(u_i)$, and a transfer amount $f$.  It also accepts a target user $u_j \in U$ and an account $ac_j \in AC(u_j)$.  

The one time transfer preforms the following operation. If $ac_i - f \ge 0$ then the interface updates $ac_i$ to $ac_i-f$ and $ac_j$ to $ac_j+f$.  Otherwise the interface does nothing.  

The simulation of one time transfer executes as follows. 

\begin{enumerate}
    \item Randomly chose inputs for the one time transfer interface.  In other words, randomly chose  a source user, $u_i \in U$, an account $ac_i \in AC(u_i)$, and a transfer amount $f$.  Also randomly chose a target user $u_j \in U$ and an account $ac_j \in AC(u_j)$.  
    \item Perform the transfer operation 
    \item Report the input that was chosen to the transfer operation and the result of the transfer operation.
\end{enumerate}

\textbf{Simulation initialization}\\
Choose the number of users to be three and the number of accounts to be 6.  Have user one own accounts 1 and 2, user two own account 3 and 4 and user 3 own account 5 and 6. Choose the amount of money in each account randomly. 

\end{humanprompt}

\begin{aireply}

This section provides an intuitive description of the system with minimal mathematical notation.

We consider a collection of users, each of whom owns a personal set of accounts. Every account belongs to exactly one user, and no account is shared among different users. Taken together, all accounts in the system are distributed across all users without overlap.

Each account stores a non-negative whole number representing the amount of money currently held in that account.

The system provides a \emph{one-time transfer} operation. This operation takes as input:
\begin{itemize}
    \item a source user and one of their accounts,
    \item a transfer amount,
    \item a target user and one of their accounts.
\end{itemize}

The operation checks whether the source account contains enough money to cover the transfer amount. If it does, the specified amount is deducted from the source account and added to the target account. If the source account does not have sufficient funds, the system makes no changes.

A simulation of this one-time transfer proceeds in three steps:
\begin{enumerate}
    \item The system randomly selects a source user, one of that user's accounts, and a transfer amount. It also randomly selects a target user and one of that user's accounts.
    \item The transfer operation is executed using these randomly chosen inputs.
    \item The system outputs both the chosen inputs and the result of the transfer attempt.
\end{enumerate}

\paragraph*{Simulation Initialization}
For the simulation, we set up three users and six accounts. User~1 owns accounts~1 and~2, User~2 owns accounts~3 and~4, and User~3 owns accounts~5 and~6. Each account is assigned an initial amount of money chosen at random.

\end{aireply}

I verified the natural language above and it captures the intent for a single transfer.  Next, I attempt to generate the natural\_math specification of the single transfer specification so that we can then obtain a proof of its correctness.  This was done in a fresh chat.  

\begin{humanprompt}

Generate a mathematical oriented precise description using latex notation of the following specification so that ambiguities are removed. 

This section provides an intuitive description of the system with minimal mathematical notation.

We consider a collection of users, each of whom owns a personal set of accounts. Every account belongs to exactly one user, and no account is shared among different users. Taken together, all accounts in the system are distributed across all users without overlap.

Each account stores a non-negative whole number representing the amount of money currently held in that account.

The system provides a \emph{one-time transfer} operation. This operation takes as input:
\begin{itemize}
    \item a source user and one of their accounts,
    \item a transfer amount,
    \item a target user and one of their accounts.
\end{itemize}

The operation checks whether the source account contains enough money to cover the transfer amount. If it does, the specified amount is deducted from the source account and added to the target account. If the source account does not have sufficient funds, the system makes no changes.

A simulation of this one-time transfer proceeds in three steps:
\begin{enumerate}
    \item The system randomly selects a source user, one of that user's accounts, and a transfer amount. It also randomly selects a target user and one of that user's accounts.
    \item The transfer operation is executed using these randomly chosen inputs.
    \item The system outputs both the chosen inputs and the result of the transfer attempt.
\end{enumerate}

\paragraph*{Simulation Initialization}
For the simulation, we set up three users and six accounts. User~1 owns accounts~1 and~2, User~2 owns accounts~3 and~4, and User~3 owns accounts~5 and~6. Each account is assigned an initial amount of money chosen at random.

\end{humanprompt}

The AI reply below is interesting.  Specifically, the natural\_math description below is \textbf{different} than the one I original wrote. 

\begin{aireply}

\paragraph{System Entities.}
Let 
\[
U = \{u_1, u_2, u_3\}
\]
be the set of users, and let
\[
AC = \{ac_1,ac_2,ac_3,ac_4,ac_5,ac_6\}
\]
be the set of accounts.  
Each account belongs to exactly one user. We formalise ownership through a function
\[
\owner: AC \to U.
\]
For the simulation initialization we define:
\[
\owner(ac_1)=\owner(ac_2)=u_1,\qquad
\owner(ac_3)=\owner(ac_4)=u_2,\qquad
\owner(ac_5)=\owner(ac_6)=u_3.
\]

Each account stores a non-negative integer balance.  
Define the balance function
\[
\bal : AC \to \mathbb{N}_0.
\]

\paragraph{One-Time Transfer Operation.}
A transfer input is a tuple
\[
(s, a_s, m, t, a_t)
\]
where
\[
s, t \in U, \qquad a_s, a_t \in AC, \qquad \owner(a_s)=s,\ \owner(a_t)=t, \qquad m \in \mathbb{N}_0.
\]

The transfer operation is the partial function
\[
T : U \times AC \times \mathbb{N}_0 \times U \times AC 
   \longrightarrow AC \to \mathbb{N}_0,
\]
which transforms the balance function as follows.  
Given state $\bal$ and input $(s,a_s,m,t,a_t)$,
\[
T(\bal; s,a_s,m,t,a_t) = 
\begin{cases}
\bal' & \text{if } \bal(a_s) \ge m,\\[2mm]
\bal & \text{otherwise},
\end{cases}
\]
where $\bal'$ is defined pointwise by:
\[
\bal'(x) =
\begin{cases}
\bal(x) - m & \text{if } x = a_s,\\[1mm]
\bal(x) + m & \text{if } x = a_t,\\[1mm]
\bal(x)     & \text{otherwise}.
\end{cases}
\]

Thus, the transfer succeeds iff the source account contains at least $m$ units of money.

\paragraph{Randomised Simulation Step.}
A single simulation step consists of selecting inputs uniformly at random and executing the transfer. Formally, the system samples:
\[
s \sim U,\qquad 
a_s \sim \{\,a \in AC : \owner(a)=s\,\},
\]
\[
t \sim U,\qquad 
a_t \sim \{\,a \in AC : \owner(a)=t\,\},
\]
and a transfer amount 
\[
m \sim M,
\]
where $M \subseteq \mathbb{N}_0$ is a predefined finite set of admissible transfer amounts.

Given the sampled inputs, the system computes
\[
\bal_{\mathrm{new}} := T(\bal; s,a_s,m,t,a_t).
\]

The output of the simulation step is the tuple
\[
(s, a_s, m, t, a_t, \bal_{\mathrm{new}}).
\]

\paragraph{Simulation Initialization.}
Initially, each account $ac_i$ is assigned a balance
\[
\bal(ac_i) \sim D,
\]
where $D$ is a specified distribution on $\mathbb{N}_0$ (e.g., uniform over a fixed finite range).

Thus the initial state is completely determined by the random vector
\[
(\bal(ac_1), \ldots, \bal(ac_6)).
\]

\end{aireply}

There are several problems with the generation above.  There are latex errors.   In addition, the specification is not stated in general but for a specific configuration.  The relation correctness is still described in english and not mathematically, e.g., the function is one on one.  It is assumed that every account has a positive amount of money which is not correct (it could be zero).  In addition the transfer can be zero and that is not correct.  Below I try to improve the prompt ina new session.  

\begin{humanprompt}

\textbf{Generate a mathematical oriented precise description using latex notation of the following specification so that ambiguities are removed.  Clearly distinguish between the specific configuration that apply to the simulation and the general specification that has arbitrary but finite number of users and accounts. }

This section provides an intuitive description of the system with minimal mathematical notation.

We consider a collection of users, each of whom owns a personal set of accounts. Every account belongs to exactly one user, and no account is shared among different users. Taken together, all accounts in the system are distributed across all users without overlap.

Each account stores a non-negative whole number representing the amount of money currently held in that account.

The system provides a \emph{one-time transfer} operation. This operation takes as input:
\begin{itemize}
    \item a source user and one of their accounts,
    \item a transfer amount,
    \item a target user and one of their accounts.
\end{itemize}

The operation checks whether the source account contains enough money to cover the transfer amount. If it does, the specified amount is deducted from the source account and added to the target account. If the source account does not have sufficient funds, the system makes no changes.

A simulation of this one-time transfer proceeds in three steps:
\begin{enumerate}
    \item The system randomly selects a source user, one of that user's accounts, and a transfer amount. It also randomly selects a target user and one of that user's accounts.
    \item The transfer operation is executed using these randomly chosen inputs.
    \item The system outputs both the chosen inputs and the result of the transfer attempt.
\end{enumerate}

\paragraph*{Simulation Initialization}
For the simulation, we set up three users and six accounts. User~1 owns accounts~1 and~2, User~2 owns accounts~3 and~4, and User~3 owns accounts~5 and~6. Each account is assigned an initial amount of money chosen at random.

\end{humanprompt}

Here is what I got with the new prompt above from the AI (using ChatGPT).  (Interestingly I had to insist that the rendering will be in Latex.)

\begin{aireply}

\subsubsection*{General Mathematical Specification}

\paragraph*{Users and Accounts}

Let \(U\) be a finite, nonempty set of users:
\[
U = \{u_1, \dots, u_n\}.
\]

Let \(AC\) be a finite, nonempty set of accounts:
\[
AC = \{ac_1, \dots, ac_k\}.
\]

Each account belongs to exactly one user.  
Ownership is represented by a total function
\[
\mathrm{owner} : AC \to U.
\]

For each user \(u \in U\), the set of accounts owned by \(u\) is
\[
AC(u) = \{\, ac \in AC \mid \mathrm{owner}(ac) = u \,\}.
\]
These sets form a partition of \(AC\):
\[
AC = \biguplus_{u \in U} AC(u).
\]

\paragraph*{Account Balances}

Each account holds a non-negative integer balance:
\[
\mathrm{bal} : AC \to \mathbb{N}.
\]

\paragraph*{One-Time Transfer Operation}

A transfer attempt is a 4-tuple
\[
(s, a_s, m, a_t)
\]
where
\begin{itemize}
    \item \(s \in U\) is the source user,
    \item \(a_s \in AC(s)\) is the source account,
    \item \(m \in \mathbb{N}\) is the transfer amount,
    \item \(a_t \in AC(t)\) is a target account for some \(t \in U\).
\end{itemize}

The transfer operation updates the balance function according to:
\[
\mathrm{bal}' =
\begin{cases}
\text{if } \mathrm{bal}(a_s) \ge m: &
\begin{aligned}
\mathrm{bal}'(a_s) &= \mathrm{bal}(a_s) - m,\\
\mathrm{bal}'(a_t) &= \mathrm{bal}(a_t) + m,\\
\mathrm{bal}'(ac) &= \mathrm{bal}(ac) \quad \text{for all } ac \notin \{a_s, a_t\},
\end{aligned}
\\[1em]
\text{if } \mathrm{bal}(a_s) < m: &
\mathrm{bal}' = \mathrm{bal}.
\end{cases}
\]

\subsubsection*{Simulation-Specific Configuration}

For the simulation, instantiate the general specification with:
\[
U_{\mathrm{sim}} = \{u_1, u_2, u_3\},
\qquad
AC_{\mathrm{sim}} = \{ac_1, ac_2, ac_3, ac_4, ac_5, ac_6\}.
\]

The ownership relation is fixed as:
\[
AC(u_1) = \{ac_1, ac_2\}, \quad
AC(u_2) = \{ac_3, ac_4\}, \quad
AC(u_3) = \{ac_5, ac_6\}.
\]

Initial account balances are chosen randomly:
\[
\mathrm{bal}(ac_i) \in \mathbb{N}
\qquad \text{for each } ac_i \in AC_{\mathrm{sim}}.
\]

\paragraph*{Simulation Procedure}

A single simulation step consists of:

\begin{enumerate}
    \item Randomly select:
    \[
        s \in U_{\mathrm{sim}}, \quad
        a_s \in AC(s), \quad
        m \in \mathbb{N}, \quad
        t \in U_{\mathrm{sim}}, \quad
        a_t \in AC(t).
    \]
    \item Execute the transfer operation on these inputs.
    \item Output both the selected inputs and the resulting balance function \(\mathrm{bal}'\).
\end{enumerate}

\end{aireply}

It worked.  In order for it to work I had to ask to distinguish the specific simulation from the general specification.  Next, we want to determine if we can get a proof from the generated natural\_math description. 

\paragraph{Assertion - the overall money in the system is constant}

\begin{humanprompt}

\begin{lemma}
\label{assertion}
The amount of money $\sum_{ac_i \in AC} ac_i$ remains constant after each fun transfer. 
\end{lemma}

\textbf{Provide a proof that lemma \ref{assertion} is correct}\\

\end{humanprompt}

Generated proof. 

\begin{aireply}

We provide a formal proof of Lemma \ref{assertion} based on the definitions given in the general specification.

\begin{proof}

Let \(\mathrm{bal} : AC \to \mathbb{N}\) be the current balance function, and consider a one-time transfer
\[
(s, a_s, m, a_t)
\]
as defined in the specification. There are two cases:

\medskip
\noindent \textbf{Case 1: Transfer is allowed (\(\mathrm{bal}(a_s) \ge m\))}

By definition, the updated balance function \(\mathrm{bal}'\) satisfies:
\[
\mathrm{bal}'(a_s) = \mathrm{bal}(a_s) - m, \qquad
\mathrm{bal}'(a_t) = \mathrm{bal}(a_t) + m, \qquad
\mathrm{bal}'(ac) = \mathrm{bal}(ac) \text{ for all } ac \notin \{a_s, a_t\}.
\]

The total money in the system after the transfer is
\[
\sum_{ac \in AC} \mathrm{bal}'(ac)
= \sum_{ac \notin \{a_s, a_t\}} \mathrm{bal}'(ac) + \mathrm{bal}'(a_s) + \mathrm{bal}'(a_t)
= \sum_{ac \notin \{a_s, a_t\}} \mathrm{bal}(ac) + (\mathrm{bal}(a_s) - m) + (\mathrm{bal}(a_t) + m).
\]

Simplifying,
\[
\sum_{ac \in AC} \mathrm{bal}'(ac)
= \sum_{ac \notin \{a_s, a_t\}} \mathrm{bal}(ac) + \mathrm{bal}(a_s) + \mathrm{bal}(a_t)
= \sum_{ac \in AC} \mathrm{bal}(ac).
\]

Hence, the total money remains constant in this case.

\medskip
\noindent \textbf{Case 2: Transfer is not allowed (\(\mathrm{bal}(a_s) < m\))}

By definition, the balance function does not change:
\[
\mathrm{bal}' = \mathrm{bal}.
\]

Therefore,
\[
\sum_{ac \in AC} \mathrm{bal}'(ac) = \sum_{ac \in AC} \mathrm{bal}(ac),
\]
so the total money remains constant.

\medskip
\noindent \textbf{Conclusion:} 

In both possible cases, the total sum of all account balances does not change. Therefore, Lemma \ref{assertion} is correct:
\[
\sum_{ac \in AC} \mathrm{bal}'(ac) = \sum_{ac \in AC} \mathrm{bal}(ac) \quad \forall \text{ one-time transfers}.
\]

\end{proof}

\end{aireply}

I reviewed the proof and it looks fine. 

\subsubsection{Natural to natural\_math to proof when the natural specification is not complete}

Next, I injected an error in the specification.  I would like to determine if the AI will spot it.   
Specifically,  I remove the check that there is enough money in the account to transfer. 

\begin{humanprompt}

\textbf{Generate a mathematical oriented precise description using latex notation of the following specification so that ambiguities are removed.  Clearly distinguish between the specific configuration that apply to the simulation and the general specification that has arbitrary but finite number of users and accounts. }

This section provides an intuitive description of the system with minimal mathematical notation.

We consider a collection of users, each of whom owns a personal set of accounts. Every account belongs to exactly one user, and no account is shared among different users. Taken together, all accounts in the system are distributed across all users without overlap.

Each account stores a non-negative whole number representing the amount of money currently held in that account.

The system provides a \emph{one-time transfer} operation. This operation takes as input:
\begin{itemize}
    \item a source user and one of their accounts,
    \item a transfer amount,
    \item a target user and one of their accounts.
\end{itemize}

The specified amount is deducted from the source account and added to the target account. 

A simulation of this one-time transfer proceeds in three steps:
\begin{enumerate}
    \item The system randomly selects a source user, one of that user's accounts, and a transfer amount. It also randomly selects a target user and one of that user's accounts.
    \item The transfer operation is executed using these randomly chosen inputs.
    \item The system outputs both the chosen inputs and the result of the transfer attempt.
\end{enumerate}

\paragraph*{Simulation Initialization}
For the simulation, we set up three users and six accounts. User~1 owns accounts~1 and~2, User~2 owns accounts~3 and~4, and User~3 owns accounts~5 and~6. Each account is assigned an initial amount of money chosen at random.

\end{humanprompt}

The AI (ChatGPT) generated the following natural\_math specification.   

\begin{aireply}

\subsubsection*{Mathematical Specification of the System}

\paragraph*{General Specification}

\paragraph{Users and Accounts.}
Let \(U\) be a finite, nonempty set of users,
\[
U = \{u_1, u_2, \ldots, u_n\},
\]
and let \(AC\) be a finite, nonempty set of accounts,
\[
AC = \{a_1, a_2, \ldots, a_m\}.
\]

\paragraph{Ownership Function.}
We define a total function
\[
\owner : AC \to U,
\]
where for each account \(a \in AC\), the value \(\owner(a)\) is the unique user who owns \(a\).
Thus, every account belongs to exactly one user, and no account is shared among different users.

\paragraph{Balances.}
A system state is a function
\[
B : AC \to \mathbb{N}.
\]

\paragraph{Transfer Operation.}
A one-time transfer is a tuple
\[
(srcU,\ srcA,\ amt,\ tgtU,\ tgtA)
\]
such that
\[
\owner(srcA) = srcU, \qquad \owner(tgtA) = tgtU, \qquad amt \in \mathbb{N}.
\]

The resulting state \(B'\) after executing the transfer on state \(B\) is
\[
B'(a) =
\begin{cases}
B(a) - amt & \text{if } a = srcA,\\[1.5mm]
B(a) + amt & \text{if } a = tgtA,\\[1.5mm]
B(a) & \text{otherwise}.
\end{cases}
\]

\paragraph*{Specific Simulation Configuration}

The simulation uses
\[
U = \{u_1, u_2, u_3\}, \qquad 
AC = \{a_1, a_2, a_3, a_4, a_5, a_6\}.
\]

The ownership function is explicitly defined by
\[
\owner(a_1) = u_1,\quad \owner(a_2) = u_1,
\]
\[
\owner(a_3) = u_2,\quad \owner(a_4) = u_2,
\]
\[
\owner(a_5) = u_3,\quad \owner(a_6) = u_3.
\]

Initial balances are chosen randomly:
\[
B(a_i) \in \mathbb{N} \quad \text{for } i = 1,\ldots,6.
\]

\end{aireply}

\paragraph{Assertion - the overall money in the system is constant}

\begin{humanprompt}

\begin{lemma}
\label{assertion}
The amount of money $\sum_{ac_i \in AC} ac_i$ remains constant after each fun transfer. 
\end{lemma}

\textbf{Provide a proof that lemma \ref{assertion} is correct or highlight why it is not.}\\

\end{humanprompt}

Following is the generated AI response that highlighted the ambiguity I intentionally injected. 

\begin{aireply}

\paragraph*{Analysis of Problems and Ambiguities}

Before proving Lemma~\ref{assertion}, we list the issues and ambiguities
that must be resolved in order for the statement to be mathematically valid.

\begin{enumerate}
    \item \textbf{Incorrect summation notation.}
    The expression
    \[
        \sum_{ac_i \in AC} ac_i
    \]
    is not meaningful: the summand should be the \emph{balance} of account
    \(ac_i\), i.e.\ \(B(ac_i)\), not the account identifier itself.

    \item \textbf{Unclear reference to ``fun transfer''.}
    The intended operation is the formally defined one-time transfer.
    We assume that it is the operation described earlier, producing a new
    state \(B'\).

    \item \textbf{Missing assumption about transfer execution.}
    The specification does not say whether a transfer is allowed even when
    \(B(srcA) < amt\), i.e.\ if balances may become negative.
    The proof below works for either case, because it only uses algebraic
    cancellation \( -amt + amt = 0 \).

    \item \textbf{Missing distinction between general state and simulation state.}
    The lemma must be stated for arbitrary \(AC\) and arbitrary states
    \(B\), not only for the fixed six-account simulation.

    \item \textbf{Incorrect type of the transferred quantity.}
    The lemma refers to “amount of money”, but mathematically this means
    the sum of the values of the balance function \(B\).
\end{enumerate}

These issues are corrected in the formal lemma and proof below.

\paragraph*{Lemma: Conservation of Total Money}

\begin{lemma}[Conservation of total money]
\label{assertion}
Let \(B : AC \to \mathbb{Z}\) be the balance function before a one-time transfer 
\((srcU,srcA,amt,tgtU,tgtA)\), and let \(B'\) be the updated balance function
obtained using the transfer semantics
\[
B'(a) =
\begin{cases}
B(a) - amt & \text{if } a = srcA,\\[1mm]
B(a) + amt & \text{if } a = tgtA,\\[1mm]
B(a)        & \text{otherwise}.
\end{cases}
\]
Then the total amount of money in the system is preserved:
\[
\sum_{a\in AC} B'(a)
\,=\, 
\sum_{a\in AC} B(a).
\]
\end{lemma}

\begin{proof}
We compute the sum after the transfer using the definition of \(B'\):
\[
\sum_{a\in AC} B'(a)
= 
\sum_{\substack{a\in AC\\ a\neq srcA,tgtA}} B(a)
   + \big(B(srcA) - amt\big)
   + \big(B(tgtA) + amt\big).
\]

Rearrange the right-hand side:
\[
\sum_{a\in AC} B'(a)
=
\sum_{a\in AC} B(a)
   - amt + amt
=
\sum_{a\in AC} B(a).
\]

Thus the total money is unchanged by the transfer.
\end{proof}

\paragraph{Remark.}
The proof does not rely on the assumption that \(B(srcA) \ge amt\).  
If your system disallows transfers that would make a balance negative,
then rejected transfers satisfy \(B' = B\), and the lemma still holds.  
If later the system includes fees, taxes, or creation/destruction of funds,
this lemma must be modified accordingly.

\end{aireply}


\subsubsection{From natural\_math to the Z formal specification notation}

\paragraph{}
To overcome attention issues when generating code from a specification, we introduce
interfaces automatically. Interface decomposition separates the specification into
well-defined interaction boundaries, reducing cognitive load and enabling local
reasoning. This relationship is illustrated in Diagram~\ref{fig:attention-interface}.

\begin{figure}[h]
\centering
\begin{tikzpicture}[
    node distance=4cm,
    every node/.style={draw, rectangle, rounded corners, align=center, minimum height=1cm},
    thick,
    <->/.style={latex-latex}
]

\node (attention) {Attention};
\node (interface) [right of=attention] {Interface\\Decomposition};

\draw[<->] (attention) -- (interface);

\end{tikzpicture}
\caption{Relationship between attention and interface decomposition}
\label{fig:attention-interface}
\end{figure}
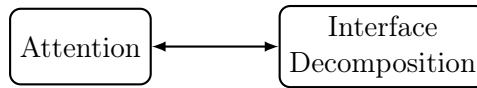

\paragraph{}
We attempt to translate the natural language specification to another intermediate representation which we name natural\_Z\_Schema. We explore the use of the \emph{Natural Logic Z Schema} representation to overcome attention issues in AI systems through decomposition into interfaces using the Z schema
mechanism (see \ref{fig:natural-z-schema}). This representation preserves as much natural language as possible to enhance readability, while using the Z schema construct to define interfaces, invariants, and
associated proofs.

\begin{figure}[h]
\centering
\begin{tikzpicture}[
    every node/.style={
        draw,
        rounded corners,
        align=center,
        minimum width=3.4cm,
        minimum height=1cm
    },
    arrow/.style={->, thick}
]

\node (natural)  at (0,0)   {Natural\\Language};
\node (logic)    at (5,0)   {Natural\\Logic};
\node (zschema)  at (5,-3)  {Natural Logic\\Z Schema};
\node (code)     at (10,0)  {Code};
\node (proof)    at (10,-3) {Proofs\\\& Invariants};

\draw[arrow] (natural) -- (logic);
\draw[arrow] (natural) |- (zschema);

\draw[arrow] (logic) -- (code);
\draw[arrow] (logic) -- (proof);
\draw[arrow] (zschema) -- (code);
\draw[arrow] (zschema) -- (proof);

\draw[arrow] (proof) -- (code);

\end{tikzpicture}
\caption{Separation of code generation parts of the specification using Z Schema}
\label{fig:natural-z-schema}
\end{figure}
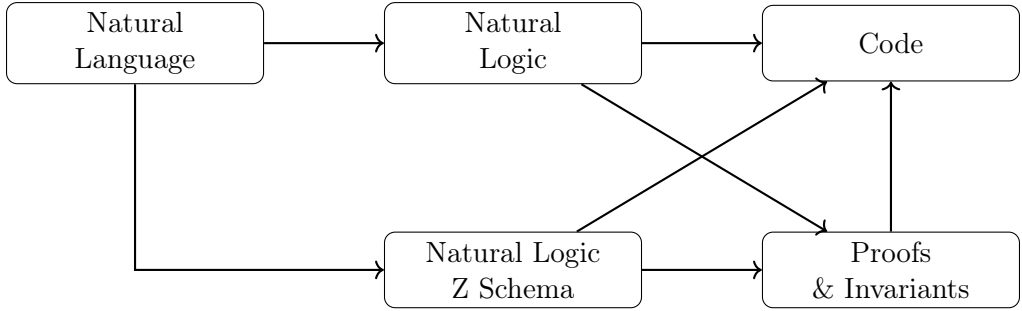

\paragraph{From natural\_math to the Z formal specification notation - First trail}

\begin{humanprompt}
    
The Z formal specification notation has the notion of schema.  A Z schema denotes a set of states or state transitions constrained by logic; at the code level it corresponds to a typed interface with invariants and operation contracts.  

\textbf{Generate a mathematical oriented precise description using latex notation of the following specification so that ambiguities are removed.  Clearly distinguish between the specific configuration that apply to the simulation and the general specification that has arbitrary but finite number of users and accounts.  Use the Z schema notation to explicitly distinguish interfaces that transfer funds and store and retrieve fund scheduled fund transfers or any other interfaces that you decide you need to define. }

This section provides an intuitive description of the system with minimal mathematical notation.

We consider a collection of users, each of whom owns a personal set of accounts. Every account belongs to exactly one user, and no account is shared among different users. Taken together, all accounts in the system are distributed across all users without overlap.

Each account stores a non-negative whole number representing the amount of money currently held in that account.

The system provides a \emph{one-time transfer} operation. This operation takes as input:
\begin{itemize}
    \item a source user and one of their accounts,
    \item a transfer amount,
    \item a target user and one of their accounts.
\end{itemize}

The operation checks whether the source account contains enough money to cover the transfer amount. If it does, the specified amount is deducted from the source account and added to the target account. If the source account does not have sufficient funds, the system makes no changes.

A simulation of this one-time transfer proceeds in three steps:
\begin{enumerate}
    \item The system randomly selects a source user, one of that user's accounts, and a transfer amount. It also randomly selects a target user and one of that user's accounts.
    \item The transfer operation is executed using these randomly chosen inputs.
    \item The system outputs both the chosen inputs and the result of the transfer attempt.
\end{enumerate}

\paragraph*{Simulation Initialization}
For the simulation, we set up three users and six accounts. User~1 owns accounts~1 and~2, User~2 owns accounts~3 and~4, and User~3 owns accounts~5 and~6. Each account is assigned an initial amount of money chosen at random.

\end{humanprompt}

\begin{aireply}

\subparagraph{Z-Style Formal Specification with Explicit Environment}

We model the system using Z-style schemas encoded in standard mathematical notation.
Let $USER$ and $ACCOUNT$ be given finite sets, and let
\[
Money \;\triangleq\; \mathbb{N}.
\]

\paragraph{General system state.}

Account ownership and balances are defined as follows:
\[
\begin{array}{l}
Ownership \;\triangleq\;
\left\{
\begin{array}{l}
owns : USER \rightharpoonup \mathcal{P}(ACCOUNT) \\
\forall u_1, u_2 \in USER,\;
u_1 \neq u_2 \Rightarrow
owns(u_1) \cap owns(u_2) = \emptyset
\end{array}
\right.
\\[1.2ex]
Balances \;\triangleq\;
\left\{
\begin{array}{l}
balance : ACCOUNT \rightharpoonup Money \\
\mathrm{dom}(balance) = \bigcup \mathrm{ran}(owns)
\end{array}
\right.
\\[1.2ex]
SystemState \;\triangleq\; Ownership \;\wedge\; Balances
\end{array}
\]

\paragraph{Environment.}

The environment nondeterministically selects transfer parameters but does not
change the system state:
\[
\begin{array}{l}
Environment \;\triangleq\;
\left\{
\begin{array}{l}
srcUser?, dstUser? \in USER \\
srcAcc?, dstAcc? \in ACCOUNT \\
amount? \in Money
\end{array}
\right.
\end{array}
\]

\paragraph{One-time transfer interface.}

The system reacts to the environment inputs. If sufficient funds exist, balances
are updated; otherwise, the state is unchanged.

\[
\begin{array}{l}
OneTimeTransfer \;\triangleq\;
\left\{
\begin{array}{l}
SystemState, SystemState' , Environment \\
srcAcc? \in owns(srcUser?) \\
dstAcc? \in owns(dstUser?) \\
amount? \le balance(srcAcc?) \\
balance' =
balance \oplus
\{
srcAcc? \mapsto balance(srcAcc?) - amount?, \\
\qquad\qquad
dstAcc? \mapsto balance(dstAcc?) + amount?
\}
\end{array}
\right.
\end{array}
\]

\[
\begin{array}{l}
NoTransfer \;\triangleq\;
\left\{
\begin{array}{l}
SystemState, SystemState', Environment \\
SystemState' = SystemState \\
srcAcc? \in owns(srcUser?) \\
dstAcc? \in owns(dstUser?) \\
amount? > balance(srcAcc?)
\end{array}
\right.
\end{array}
\]

\[
Transfer \;\triangleq\; OneTimeTransfer \;\lor\; NoTransfer
\]

\paragraph{Observation interface.}

The simulation observes the selected inputs and whether the transfer succeeded.
\[
\begin{array}{l}
TransferRecord \;\triangleq\;
\left\{
\begin{array}{l}
srcUser, dstUser \in USER \\
srcAcc, dstAcc \in ACCOUNT \\
amount \in Money \\
success \in \mathbb{B}
\end{array}
\right.
\end{array}
\]

\[
\begin{array}{l}
SimulationStep \;\triangleq\;
\left\{
\begin{array}{l}
SystemState, SystemState', Environment, record! \\
Transfer \\
record!.srcUser = srcUser? \\
record!.srcAcc  = srcAcc? \\
record!.dstUser = dstUser? \\
record!.dstAcc  = dstAcc? \\
record!.amount  = amount? \\
record!.success \Leftrightarrow amount? \le balance(srcAcc?)
\end{array}
\right.
\end{array}
\]

\paragraph{Concrete simulation configuration.}

The simulation instantiates the abstract specification with fixed sets:
\[
USER = \{u_1,u_2,u_3\}, \qquad
ACCOUNT = \{a_1,a_2,a_3,a_4,a_5,a_6\}
\]

with ownership
\[
owns =
\{
u_1 \mapsto \{a_1,a_2\},
u_2 \mapsto \{a_3,a_4\},
u_3 \mapsto \{a_5,a_6\}
\},
\]
and balances initialized nondeterministically subject to
$\mathrm{dom}(balance)=ACCOUNT$.

\end{aireply}

The definition above is correct but the notation became cumbersome.  I'll make a second attempt below to obtain a that uses the Z notation only for the interfaces definitions using the schema mechanism.  

\paragraph{From natural\_math to the Z formal specification notation - Second trail}

\begin{humanprompt}

The Z formal specification notation has the notion of schema.  A Z schema denotes a set of states or state transitions constrained by logic; at the code level it corresponds to a typed interface with invariants and operation contracts.  

\textbf{ Modify the following specification to use the Z schema notation to explicitly distinguish interfaces that transfer funds and store and retrieve fund scheduled fund transfers.  Leave other aspects of the specification and do not translate them to Z. }

Here is the specification \ref{scheudleTransfer}.
    
\end{humanprompt}

This step required several iterations.  Essentially had to specify that my original natural\_logic specification should "call" the new defined schemas.

\begin{aireply}

The Z formal specification notation has the notion of a \emph{schema}. A Z schema denotes a set of states or state transitions constrained by logic; at the code level it corresponds to a typed interface with invariants and operation contracts.

The following is a single, clean specification written in \emph{plain Z notation}.  
No Z libraries, schema environments, or special packages are assumed.  
The original text is minimally edited so that it explicitly refers to the relevant Z schemas, which are defined using standard mathematical layout only.

\subparagraph{System model}

We are given a set of users
\[
U = \{ u_1, \ldots, u_n\}.
\]
Each user has a set of accounts that belong to that user. We have a set of accounts
\[
AC = \{ac_1,\ldots, ac_k\}
\]
and a function
\[
AC(\cdot) : U \rightarrow \mathcal{P}(AC).
\]

For any user $u \in U$, $AC(u) \subseteq AC$ is the subset of accounts that belong to $u$.

We assume that for any two users $u_i, u_j \in U$,
\[
AC(u_i) \cap AC(u_j) = \emptyset
\]
and
\[
\bigcup_{u \in U} AC(u) = AC.
\]
Thus, each account belongs to exactly one user.

Each account $ac \in AC$ has a balance represented by a natural number (including zero).  
Account balances are modeled by the Z state schema \emph{AccountState} defined below.

Fund transfers are a function of the current time. The current time is a natural number initialized to $0$ at the beginning of the simulation. At each stage of the simulation the current time is first incremented by one.

\subparagraph{Fund transfer operations}

\paragraph{One-time transfer.}
The one-time transfer is an interface that accepts a source user $u_i \in U$, a source account $ac_i \in AC(u_i)$, a target user $u_j \in U$, a target account $ac_j \in AC(u_j)$, and a transfer amount $f$.

The semantics of this interface are defined by the Z schemas \emph{OneTimeTransfer} and \emph{OneTimeTransferFail}. If the source account has sufficient funds, balances are updated; otherwise, the state is unchanged.

\paragraph{Planned transfer.}
A planned transfer is a one-time transfer scheduled to occur at a future time $t$.

Planned transfers are stored in a planned transfer set $P$, modeled by the Z state schema \emph{PlannedTransferStore}.  
Scheduling a planned transfer is defined by the Z schema \emph{SchedulePlannedTransfer}.  
Execution of planned transfers due at the current time is defined by the Z schema \emph{ExecutePlannedTransfers}.

\subparagraph{Simulation}

The simulation of fund transfers consists of stages. In each stage the following steps are taken.

\begin{enumerate}
  \item Increment the current time $t$ by $1$.
  \item \textbf{One-time transfer.} Randomly choose inputs and perform the one-time transfer using the interface defined by \emph{OneTimeTransfer} and \emph{OneTimeTransferFail}.
  \item \textbf{Schedule a planned transfer.} Randomly choose inputs for a one-time transfer and a time $t' > t$. Add the planned transfer to $P$ using the interface defined by \emph{SchedulePlannedTransfer}.
  \item \textbf{Perform planned transfers.} For each planned transfer scheduled at time $t$, execute it using the interface defined by \emph{ExecutePlannedTransfers}.
  \item Report the chosen inputs, the result of each transfer, the transfer type (one-time or planned), the current time $t$, and the state of $P$ before planned transfers are executed.
\end{enumerate}

\subparagraph{Simulation initialization}

Choose three users and six accounts. User one owns accounts 1 and 2, user two owns accounts 3 and 4, and user three owns accounts 5 and 6. Choose the initial balance of each account randomly.

The current time is initialized to $0$. The planned transfer set $P$ is initialized to the empty set. Each element of $P$ has type
\[
U \times AC \times U \times AC \times \mathbb{N}^+ \times \mathbb{N}.
\]

\subparagraph{Z schemas (plain Z notation)}

\paragraph{AccountState}
\[
AccountState \;\defs\;
\left[
\begin{array}{l}
balances : AC \rightarrow \mathbb{N}
\end{array}
\right]
\]

\paragraph{PlannedTransfer}
\[
PlannedTransfer \;\defs\;
U \times AC \times U \times AC \times \mathbb{N}^+ \times \mathbb{N}
\]

\paragraph{PlannedTransferStore}
\[
PlannedTransferStore \;\defs\;
\left[
\begin{array}{l}
P : \mathcal{P}(PlannedTransfer)
\end{array}
\right]
\]

\paragraph{OneTimeTransfer}
\[
OneTimeTransfer \;\defs\;
\left[
\begin{array}{l}
\Delta AccountState \\
srcU, dstU : U \\
srcAC, dstAC : AC \\
f : \mathbb{N}^+ \\
\hline
balances(srcAC) \ge f \\
balances' = balances \oplus
\{
srcAC \mapsto balances(srcAC) - f,\;
dstAC \mapsto balances(dstAC) + f
\}
\end{array}
\right]
\]

\paragraph{OneTimeTransferFail}
\[
OneTimeTransferFail \;\defs\;
\left[
\begin{array}{l}
\Xi AccountState \\
srcU, dstU : U \\
srcAC, dstAC : AC \\
f : \mathbb{N}^+ \\
\hline
balances(srcAC) < f
\end{array}
\right]
\]

\paragraph{SchedulePlannedTransfer}
\[
SchedulePlannedTransfer \;\defs\;
\left[
\begin{array}{l}
\Delta PlannedTransferStore \\
srcU, dstU : U \\
srcAC, dstAC : AC \\
f : \mathbb{N}^+ \\
t' : \mathbb{N} \\
\hline
P' = P \cup \{ (srcU, srcAC, dstU, dstAC, f, t') \}
\end{array}
\right]
\]

\paragraph{ExecutePlannedTransfers}
\[
ExecutePlannedTransfers \;\defs\;
\left[
\begin{array}{l}
\Delta AccountState \\
\Delta PlannedTransferStore \\
t : \mathbb{N} \\
\hline
\forall p \in P @ \\
\quad p_6 = t \land balances(p_2) \ge p_5 \implies \\
\quad balances' =
balances \oplus
\{
p_2 \mapsto balances(p_2) - p_5,\;
p_4 \mapsto balances(p_4) + p_5
\} \\
P' = \{ p \in P \mid p_6 \neq t \}
\end{array}
\right]
\]

\end{aireply}

There are two problems with the above specification.  One is that the ExecutePlannedTransfers schema is described for the specific configuration.  The other issue is that the calls to the schema do not describe the inputs explicitly.  

\paragraph{From natural\_math to the Z formal specification notation - Third trail}

\begin{humanprompt}
Please render again fixing two problems.  One is that when stating that you will call a schema in my original specification you need to specify the inputs the schema will take.  The other is that ExecutePlannedTransfer should be described in general and not for the specific simulation configuration.  
\end{humanprompt}

\begin{aireply}

The Z formal specification notation has the notion of a \emph{schema}. A Z schema denotes a set of states or state transitions constrained by logic; at the code level it corresponds to a typed interface with invariants and operation contracts.

The following is a single, clean specification written in \emph{plain Z notation} (no Z libraries assumed).  
The original text is minimally edited so that:
\begin{itemize}
  \item whenever the text refers to a Z schema, the \emph{inputs of that schema are explicitly stated}, and
  \item the schema \emph{ExecutePlannedTransfers} is defined \emph{independently of any particular simulation configuration}.
\end{itemize}

\paragraph{System model}

We are given a set of users
\[
U = \{ u_1, \ldots, u_n\}.
\]
Each user has a set of accounts that belong to that user. We have a set of accounts
\[
AC = \{ac_1,\ldots, ac_k\}
\]
and a function
\[
AC(\cdot) : U \rightarrow \mathcal{P}(AC).
\]

For any user $u \in U$, $AC(u) \subseteq AC$ is the subset of accounts that belong to $u$.

We assume that for any two users $u_i, u_j \in U$,
\[
AC(u_i) \cap AC(u_j) = \emptyset
\]
and
\[
\bigcup_{u \in U} AC(u) = AC.
\]
Thus, each account belongs to exactly one user.

Each account $ac \in AC$ has a balance represented by a natural number (including zero).  
Account balances are modeled by the Z state schema \emph{AccountState(balances)}.

Fund transfers are a function of the current time. The current time is a natural number initialized to $0$ at the beginning of the simulation. At each stage of the simulation the current time is first incremented by one.

\subparagraph{Fund transfer operations}

\paragraph{One-time transfer.}
The one-time transfer is an interface that accepts as inputs:
\[
(u_i, ac_i, u_j, ac_j, f)
\]
where $u_i, u_j \in U$, $ac_i \in AC(u_i)$, $ac_j \in AC(u_j)$, and $f \in \mathbb{N}^+$.

Given these inputs, the operation is defined by invoking either the schema
\[
OneTimeTransfer(srcU{=}u_i, srcAC{=}ac_i, dstU{=}u_j, dstAC{=}ac_j, f)
\]
or, when insufficient funds exist, the schema
\[
OneTimeTransferFail(srcU{=}u_i, srcAC{=}ac_i, dstU{=}u_j, dstAC{=}ac_j, f).
\]

\paragraph{Planned transfer.}
A planned transfer is a one-time transfer scheduled to occur at a future time $t$.

Scheduling a planned transfer accepts as inputs:
\[
(u_i, ac_i, u_j, ac_j, f, t')
\]
with $t' > t$, and is defined by invoking the schema
\[
SchedulePlannedTransfer(srcU{=}u_i, srcAC{=}ac_i, dstU{=}u_j, dstAC{=}ac_j, f, t').
\]

Planned transfers are stored in the planned transfer set $P$, which is modeled by the Z state schema \emph{PlannedTransferStore(P)}.

Execution of planned transfers due at the current time $t$ is defined by invoking the schema
\[
ExecutePlannedTransfers(t),
\]
which examines all planned transfers in $P$ scheduled for time $t$ and applies the corresponding one-time transfers.

\subparagraph{Simulation}

The simulation of fund transfers consists of stages. In each stage the following steps are taken.

\begin{enumerate}
  \item Increment the current time $t$ by $1$.
  \item \textbf{One-time transfer.} Randomly choose inputs $(u_i, ac_i, u_j, ac_j, f)$ and invoke the schema
  \[
  OneTimeTransfer(srcU{=}u_i, srcAC{=}ac_i, dstU{=}u_j, dstAC{=}ac_j, f)
  \]
  or \emph{OneTimeTransferFail} if the precondition is not met.
  \item \textbf{Schedule a planned transfer.} Randomly choose inputs $(u_i, ac_i, u_j, ac_j, f, t')$ with $t' > t$ and invoke
  \[
  SchedulePlannedTransfer(srcU{=}u_i, srcAC{=}ac_i, dstU{=}u_j, dstAC{=}ac_j, f, t').
  \]
  \item \textbf{Perform planned transfers.} Invoke
  \[
  ExecutePlannedTransfers(t)
  \]
  to execute all planned transfers scheduled at the current time.
  \item Report the chosen inputs, the result of each transfer, the transfer type (one-time or planned), the current time $t$, and the state of $P$ immediately before planned transfers are executed.
\end{enumerate}

\subparagraph{Simulation initialization}

Choose three users and six accounts. User one owns accounts 1 and 2, user two owns accounts 3 and 4, and user three owns accounts 5 and 6. Choose the initial balance of each account randomly.

The current time is initialized to $0$. The planned transfer set $P$ is initialized to the empty set. Each element of $P$ has type
\[
U \times AC \times U \times AC \times \mathbb{N}^+ \times \mathbb{N}.
\]

\subparagraph{Z schemas (plain Z notation)}

\paragraph{AccountState}
\[
AccountState \;\defs\;
\left[
\begin{array}{l}
balances : AC \rightarrow \mathbb{N}
\end{array}
\right]
\]

\paragraph{PlannedTransfer}
\[
PlannedTransfer \;\defs\;
U \times AC \times U \times AC \times \mathbb{N}^+ \times \mathbb{N}
\]

\paragraph{PlannedTransferStore}
\[
PlannedTransferStore \;\defs\;
\left[
\begin{array}{l}
P : \mathcal{P}(PlannedTransfer)
\end{array}
\right]
\]

\paragraph{OneTimeTransfer}
\[
OneTimeTransfer \;\defs\;
\left[
\begin{array}{l}
\Delta AccountState \\
srcU, dstU : U \\
srcAC, dstAC : AC \\
f : \mathbb{N}^+ \\
\hline
balances(srcAC) \ge f \\
balances' = balances \oplus
\{
srcAC \mapsto balances(srcAC) - f,\;
dstAC \mapsto balances(dstAC) + f
\}
\end{array}
\right]
\]

\paragraph{OneTimeTransferFail}
\[
OneTimeTransferFail \;\defs\;
\left[
\begin{array}{l}
\Xi AccountState \\
srcU, dstU : U \\
srcAC, dstAC : AC \\
f : \mathbb{N}^+ \\
\hline
balances(srcAC) < f
\end{array}
\right]
\]

\paragraph{SchedulePlannedTransfer}
\[
SchedulePlannedTransfer \;\defs\;
\left[
\begin{array}{l}
\Delta PlannedTransferStore \\
srcU, dstU : U \\
srcAC, dstAC : AC \\
f : \mathbb{N}^+ \\
t' : \mathbb{N} \\
\hline
P' = P \cup \{ (srcU, srcAC, dstU, dstAC, f, t') \}
\end{array}
\right]
\]

\paragraph{ExecutePlannedTransfers}
\[
ExecutePlannedTransfers \;\defs\;
\left[
\begin{array}{l}
\Delta AccountState \\
\Delta PlannedTransferStore \\
t : \mathbb{N} \\
\hline
\forall (u_i, ac_i, u_j, ac_j, f, t') \in P @ \\
\quad t' = t \land balances(ac_i) \ge f \implies \\
\quad balances' =
balances \oplus
\{
ac_i \mapsto balances(ac_i) - f,\;
ac_j \mapsto balances(ac_j) + f
\} \\
P' = \{ p \in P \mid \pi_6(p) \neq t \}
\end{array}
\right]
\]

\end{aireply}

The AI addressed the two issues correctly. 


\paragraph{Invariant proof using the Z modeling of the interfaces}

Next we attempt to proof the invariants that were prove above using the Z modeling of the interfaces. 

\begin{humanprompt}

Given some legal configuration of the system and consider the inductive set of configurations created by the applications of the schema defined above, namely, OneTimeTransfer, OneTimeTransferFail, SchedulePlannedTransfer, and ExecutePlannedTransfers.  We refer to that inductive set as I.  Prove the following claim. 

\begin{lemma}
\label{assertion}
The amount of money $\sum_{ac_i \in AC} ac_i$ remains constant for all configurations in the inductive set I.  
\end{lemma}

\end{humanprompt}

\paragraph{Stronger invariant claim.}
Note that the following claim and its proof are strictly stronger than those presented in the previous sections.
We show that the inductively defined set \( I(C,\mathit{SCHEMA}) \), obtained from any legal initial
configuration \( C \) by repeated application of the schemas in the specification (denoted
by \( \mathit{SCHEMA} \)), satisfies the invariant claim.
In other words, every configuration reachable by any finite sequence of applications of the
operations defined by the schemas preserves the invariant.

As a consequence, the proof is not tied to any particular simulation or execution strategy.
It applies uniformly to \emph{all} possible sequences of schema applications starting from a
legal configuration.
This establishes the invariant as a global property of the specification itself, rather than
a property of a specific operational interpretation or execution model.

\begin{aireply}

Below is a Z-style invariant proof formatted so that displayed lines are short and fit comfortably within a page.

\paragraph*{Lemma}

Let \( I \) be the inductive set of configurations generated from an initial
legal configuration by finite applications of the schemas
\[
OneTimeTransfer,\;
OneTimeTransferFail,\;
SchedulePlannedTransfer,\;
ExecutePlannedTransfers.
\]

\begin{lemma}
\label{assertion}
For every configuration in \( I \), the total amount of money
\[
\sum_{ac \in AC} balances(ac)
\]
remains constant.
\end{lemma}

\subparagraph*{Proof}

The proof proceeds by induction on the construction of the inductive set \( I \).

\paragraph{Base case.}
Let the initial configuration be given by
\[
AccountState(balances_0),
\qquad
balances_0 : AC \rightarrow \mathbb{N}.
\]
The quantity
\[
\sum_{ac \in AC} balances_0(ac)
\]
is fixed. Hence the invariant holds initially.

\paragraph{Inductive hypothesis.}
Assume that for some configuration in \( I \),
\[
\sum_{ac \in AC} balances(ac) = M,
\]
for a constant \( M \).

\paragraph{Inductive step.}
We show that the invariant is preserved by each schema.

\subparagraph{Case 1: \textnormal{OneTimeTransfer}.}
The schema defines
\[
balances' =
balances \oplus
\{
srcAC \mapsto balances(srcAC) - f,
\;
dstAC \mapsto balances(dstAC) + f
\},
\]
with the precondition \( balances(srcAC) \ge f \).
All other balances are unchanged. Therefore,
\[
\sum_{ac \in AC} balances'(ac)
=
\sum_{ac \in AC} balances(ac)
- f + f
=
\sum_{ac \in AC} balances(ac).
\]

\subparagraph{Case 2: \textnormal{OneTimeTransferFail}.}
This schema includes \( \Xi AccountState \), hence
\[
balances' = balances.
\]
The invariant is preserved.

\subparagraph{Case 3: \textnormal{SchedulePlannedTransfer}.}
This schema modifies only the planned transfer store and leaves
the account state unchanged:
\[
balances' = balances.
\]
The invariant is preserved.

\subparagraph{Case 4: \textnormal{ExecutePlannedTransfers}.}
Each executed planned transfer updates balances as follows:
\[
\begin{aligned}
balances(ac_i) &\mapsto balances(ac_i) - f, \\
balances(ac_j) &\mapsto balances(ac_j) + f,
\end{aligned}
\]
subject to \( balances(ac_i) \ge f \).
Each update preserves the total sum of balances. Since
\textnormal{ExecutePlannedTransfers} is a finite composition of
such updates, and updates to the planned transfer store do not
affect balances, the invariant is preserved.

\paragraph{Conclusion.}
All schemas generating \( I \) preserve the total balance invariant.
Since it holds in the initial configuration and is preserved by every
inductive step, it holds for all configurations in \( I \):
\[
\forall c \in I @
\sum_{ac \in AC} balances(ac) = \text{constant}.
\]
\hfill\(\Box\)

\end{aireply}

\paragraph{Have the AI suggest the interfaces}

In this subsection we attempt to have the AI create the interfaces with no inhints on their possible content.  

\begin{humanprompt}
 
The Z formal specification notation has the notion of a schema.  A Z schema denotes a set of states or state transitions constrained by logic; at the code level it corresponds to a typed interface with invariants and operation contracts.  

\textbf{ Modify the following specification to use the Z schema notation to explicitly distinguish interfaces that apply to the specification below.  Leave other aspects of the specification and do not translate them to Z. }

Here is the specification \ref{scheudleTransfer}

\end{humanprompt}

Here is the AI reply. 

\begin{aireply}

The specification below is minimally augmented with
\emph{Z schema notation} in a package-free form.
Schemas are rendered using standard math displays,
making interfaces, invariants, and contracts explicit.
All other aspects remain informal and unchanged.

\bigskip

\textbf{Basic Sets and Types}

\[
[U,\; AC]
\]

\[
N = \mathbb{N}
\qquad
N^{+} = \mathbb{N} \setminus \{0\}
\]

\[
Transfer =
U \times AC \times U \times AC \times N^{+} \times N
\]

\bigskip

\textbf{Global State Interface}

\[
\begin{array}{l}
FundTransferState \;\;\widehat{=} \\[2mm]
\quad
\left[
\begin{array}{l}
accounts : AC \rightarrow N \\
owner    : AC \rightarrow U \\
time     : N \\
P        : \mathcal{P}(Transfer)
\end{array}
\;\middle|\;
\begin{array}{l}
\forall ac_1, ac_2 : AC \cdot \\
\;\; owner(ac_1) = owner(ac_2) \\
\;\; \Rightarrow ac_1 = ac_2
\end{array}
\right]
\end{array}
\]

\bigskip

\textbf{One-Time Transfer Interface}

\[
\begin{array}{l}
OneTimeTransfer \;\;\widehat{=} \\[2mm]
\quad
\left[
\begin{array}{l}
\Delta FundTransferState \\
u_s, u_t : U \\
ac_s, ac_t : AC \\
f : N^{+}
\end{array}
\;\middle|\;
\begin{array}{l}
owner(ac_s) = u_s \\
owner(ac_t) = u_t
\end{array}
\right]
\end{array}
\]

\medskip

\textbf{Transfer Contract}

If sufficient funds exist:
\[
accounts(ac_s) \ge f
\]

\[
accounts' =
accounts \oplus
\{
ac_s \mapsto accounts(ac_s) - f,\;
ac_t \mapsto accounts(ac_t) + f
\}
\]

Otherwise:
\[
accounts(ac_s) < f
\;\Rightarrow\;
accounts' = accounts
\]

\bigskip

\textbf{Planned Transfer Scheduling Interface}

\[
\begin{array}{l}
SchedulePlannedTransfer \;\;\widehat{=} \\[2mm]
\quad
\left[
\begin{array}{l}
\Delta FundTransferState \\
u_s, u_t : U \\
ac_s, ac_t : AC \\
f : N^{+} \\
t_p : N
\end{array}
\;\middle|\;
\begin{array}{l}
owner(ac_s) = u_s \\
owner(ac_t) = u_t \\
t_p > time
\end{array}
\right]
\end{array}
\]

\[
P' =
P \cup
\{(u_s, ac_s, u_t, ac_t, f, t_p)\}
\]

\bigskip

\textbf{Planned Transfer Execution Interface}

\[
\begin{array}{l}
ExecutePlannedTransfer \;\;\widehat{=} \\[2mm]
\quad
\left[
\begin{array}{l}
\Delta FundTransferState \\
p : Transfer
\end{array}
\;\middle|\;
\begin{array}{l}
p \in P \\
\pi_6(p) = time
\end{array}
\right]
\end{array}
\]

\[
P' = P \setminus \{ p \}
\]

The balance update obeys the
\texttt{OneTimeTransfer} contract.

\bigskip

\textbf{Simulation Clock Interface}

\[
\begin{array}{l}
AdvanceTime \;\;\widehat{=} \\[2mm]
\quad
\left[
\begin{array}{l}
\Delta FundTransferState
\end{array}
\;\middle|\;
\begin{array}{l}
time' = time + 1
\end{array}
\right]
\end{array}
\]

\bigskip

\textbf{Summary}

\begin{itemize}
  \item Z-style schemas are expressed using plain LaTeX.
  \item Interfaces, invariants, and contracts are explicit.
  \item No additional packages are required.
  \item The original informal specification is preserved.
\end{itemize}

\end{aireply}

On the positive, side the AI was able to identify the major interfaces by itself.  On the negative side, it mixed schema with natural language description which is confusing.  It also made the mistake of referring to $\pi_6(p)$ which is a specific simulation detail.  It also took away details such as the simulation initialization.

\end{document}